\documentclass{article}

\usepackage[T1]{fontenc}

\usepackage[english]{babel}
\usepackage{tocloft}
\usepackage[letterpaper,top=3cm,bottom=3cm,left=3.5cm,right=3.5cm,marginparwidth=1.75cm]{geometry}

\usepackage[sorting=none]{biblatex}
\usepackage{amsmath,amsfonts,amssymb}
\usepackage{graphicx}
\usepackage[colorlinks=true, allcolors=blue]{hyperref}
\usepackage{braket}
\usepackage[official]{eurosym}
\usepackage{mathtools}
\usepackage[dvipsnames]{xcolor}
\usepackage{braket}
\usepackage{authblk}
\usepackage{tcolorbox}
\usepackage[style=S,underline=true, thickness=1.0pt]{thmbox}
\usepackage{enumerate}
\usepackage[framemethod=TikZ,xcolor=RGB]{mdframed}

\definecolor{mybrown}{RGB}{102,101,71}
\definecolor{myyellow}{RGB}{255,226,138}
\definecolor{mygreen}{RGB}{111,203,159}

\hypersetup{
    colorlinks=true,
    linkcolor=blue,
    filecolor=magenta,      
    urlcolor=cyan,
    pdftitle={Overleaf Example},
    pdfpagemode=FullScreen,
    }

\newcommand{\R}{\mathtt{R}}
\renewcommand{\P}{\mathtt{P}}
\renewcommand{\S}{\mathtt{S}}

\newcommand{\Sb}{\textbf{S}}

\newcommand{\Eb}{\textbf{E}}
\newcommand{\Mb}{\textbf{M}}

\newcommand{\supp}[1]{\text{supp}\,#1}
\newcommand{\half}{\scalebox{0.9}{$\frac{1}{2}$}}
\renewcommand{\vec}{\boldsymbol}

\newcounter{theo}[section]
\renewcommand{\thetheo}{\arabic{section}.\arabic{theo}}
\newenvironment{theo}[2][]{%
\refstepcounter{theo}%
\ifstrempty{#1}%
{\mdfsetup{%
frametitle={%
\tikz[baseline=(current bounding box.east),outer sep=0pt]
\node[anchor=east,rectangle,fill=mygreen!40]
{\strut Theorem~\thetheo};}}
}%
{\mdfsetup{%
frametitle={%
\tikz[baseline=(current bounding box.east),outer sep=0pt]
\node[anchor=east,rectangle,fill=mygreen!40]
{\strut Theorem~\thetheo:~#1};}}%
}%
\mdfsetup{innertopmargin=10pt,linecolor=mygreen!40,%
linewidth=2pt,topline=true,%
frametitleaboveskip=\dimexpr-\ht\strutbox\relax
}
\begin{mdframed}[]\relax%
\label{#2}}{\end{mdframed}}

\newcounter{lem}[section]
\renewcommand{\thelem}{\arabic{section}.\arabic{lem}}
\newenvironment{lem}[2][]{%
\refstepcounter{lem}%
\ifstrempty{#1}%
{\mdfsetup{%
frametitle={%
\tikz[baseline=(current bounding box.east),outer sep=0pt]
\node[anchor=east,rectangle,fill=mygreen!40]
{\strut Lemma~\thelem};}}
}%
{\mdfsetup{%
frametitle={%
\tikz[baseline=(current bounding box.east),outer sep=0pt]
\node[anchor=east,rectangle,fill=mygreen!40]
{\strut Lemma~\thetheo:~#1};}}%
}%
\mdfsetup{innertopmargin=10pt,linecolor=mygreen!40,%
linewidth=2pt,topline=true,%
frametitleaboveskip=\dimexpr-\ht\strutbox\relax
}
\begin{mdframed}[]\relax%
\label{#2}}{\end{mdframed}}

\newcounter{defo}[section]
\renewcommand{\thedefo}{\arabic{section}.\arabic{defo}}
\newenvironment{defo}[2][]{%
\refstepcounter{defo}%
\ifstrempty{#1}%
{\mdfsetup{%
frametitle={%
\tikz[baseline=(current bounding box.east),outer sep=0pt]
\node[anchor=east,rectangle,fill=mybrown!20]
{\strut Definition~\thedefo};}}
}%
{\mdfsetup{%
frametitle={%
\tikz[baseline=(current bounding box.east),outer sep=0pt]
\node[anchor=east,rectangle,fill=mybrown!20]
{\strut Definition~\thetheo:~#1};}}%
}%
\mdfsetup{innertopmargin=10pt,linecolor=mybrown!20,%
linewidth=2pt,topline=true,%
frametitleaboveskip=\dimexpr-\ht\strutbox\relax
}
\begin{mdframed}[]\relax%
\label{#2}}{\end{mdframed}}

\newcounter{probo}[section]
\renewcommand{\theprobo}{\arabic{section}.\arabic{probo}}
\newenvironment{probo}[2][]{%
\refstepcounter{probo}%
\ifstrempty{#1}%
{\mdfsetup{%
frametitle={%
\tikz[baseline=(current bounding box.east),outer sep=0pt]
\node[anchor=east,rectangle,fill=myyellow!50]
{\strut Problem~\theprobo};}}
}%
{\mdfsetup{%
frametitle={%
\tikz[baseline=(current bounding box.east),outer sep=0pt]
\node[anchor=east,rectangle,fill=myyellow!50]
{\strut Problem~\thetheo:~#1};}}%
}%
\mdfsetup{innertopmargin=10pt,linecolor=myyellow!50,%
linewidth=2pt,topline=true,%
frametitleaboveskip=\dimexpr-\ht\strutbox\relax
}
\begin{mdframed}[]\relax%
\label{#2}}{\end{mdframed}}

\title{\textbf{Contextuality and inductive bias in quantum machine learning}}

\author[1]{Joseph Bowles\thanks{joseph@xanadu.ai}}
\author[2]{Victoria J Wright}
\author[2]{M\'{a}t\'{e} Farkas}
\author[1]{\\ Nathan Killoran}
\author[1]{Maria Schuld}

\affil[1]{\small{Xanadu, Toronto, ON, M5G 2C8, Canada}}
\affil[2]{\small{ICFO-Institut de Ciencies Fotoniques, The Barcelona Institute of Science and Technology, 08860 Castelldefels, Spain}}

\date{}

\begin{document}
\maketitle

\begin{abstract}
    Generalisation in machine learning often relies on the ability to encode structures present in data into an inductive bias of the model class. To understand the power of quantum machine learning, it is therefore crucial to identify the types of data structures that lend themselves naturally to quantum models. In this work we look to quantum contextuality---a form of nonclassicality with links to computational advantage---for answers to this question. We introduce a framework for studying contextuality in machine learning, which leads us to a definition of what it means for a learning model to be contextual. From this, we connect a central concept of  contextuality, called operational equivalence, to the ability of a model to encode a linearly conserved quantity in its label space. A consequence of this connection is that contextuality is tied to expressivity: contextual model classes that encode the inductive bias are generally more expressive than their noncontextual counterparts. To demonstrate this, we construct an explicit toy learning problem---based on learning the payoff behaviour of a zero-sum game---for which this is the case. By leveraging tools from geometric quantum machine learning, we then describe how to construct quantum learning models with the associated inductive bias, and show through our toy problem that they outperform their corresponding classical surrogate models. This suggests that understanding learning problems of this form may lead to useful insights about the power of quantum machine learning.
\end{abstract}

\section{Introduction}
In order for a learning model to generalise well from training data, it is often crucial to encode some knowledge about the structure of the data into the model itself \cite{ibml1,ibml2,ibml3}. Convolutional neural networks \cite{cnn0,cnn1,cnn2} are a classic illustration of this principle, whose success at image related tasks is often credited to the existence of model structures that relate to label invariance of the data under translation symmetries \cite{bronstein2017geometric}. Together with the choice of loss function and hyperparameters, these structures form part of the basic assumptions that a learning model makes about the data, which is commonly referred to as the \emph{inductive bias} of the model. 

One of the central challenges facing quantum machine learning is to identify data structures that can be encoded usefully into quantum learning models; in other words, what are the forms of inductive bias that naturally lend themselves to quantum computation \cite{ibq1,ibq2,ibq3,surrogates}? In answering this question, we should be wary of hoping for a one-size-fits-all approach in which quantum models outperform neural network models at generic learning tasks. Rather, effort should be placed in understanding how the Hilbert space structure and probabilistic nature of the theory suggest particular biases for which quantum machine learning may excel. Indeed, an analogous perspective is commonplace in quantum computation, where computational advantages are expected only for specific problems that happen to benefit from the peculiarities of quantum logic. 

In the absence of large quantum computers and in the infancy of quantum machine learning theory, how should we look for insight on this issue? One possibility is to turn to complexity theory \cite{qcomplex1,qcomplex2}, where asymptotic advantages of quantum learning algorithms have been proven \cite{speedup,speedup2,speedup3,speedup4}. These results are few and far between however, and the enormous gap between what is possible to prove in a complexity-theoretic sense, and the types of advantages that may be possible in practice, means that there are growing doubts about the practical relevance of these results. Indeed, progress in machine learning is often the result of good ideas built on intuition, rather than worst-case complexity theoretic analysis. To repeat a common quip: many problems in machine learning are NP-hard, but neural networks don’t know that so they solve them anyway. 

We will take a different route, and lean on the field of quantum foundations to guide us. Quantum foundations is predominantly concerned with understanding the frontier between the quantum and classical world, and typically values a clear qualitative understanding of a phenomenon over purely mathematical knowledge. For these reasons it is well suited to identify features of quantum theory that may advance quantum machine learning in useful directions. In particular, we focus on the phenomenon of contextuality \cite{con1,con2,gencon}, which is perhaps the most prominent form of nonclassicality studied in the literature. Contextuality has a considerable tradition of being studied in relation to quantum computation \cite{concom1,concom2,concom3,concom4,concom5,concom6,needcon}, where it is closely connected to the possibility of computational speed-up. Despite this, it has had relatively little attention in quantum machine learning, with only a couple of works \cite{conml1,conml2} linking contextuality to implications for learning. 

\begin{figure}
    \centering
    \includegraphics[width=\textwidth]{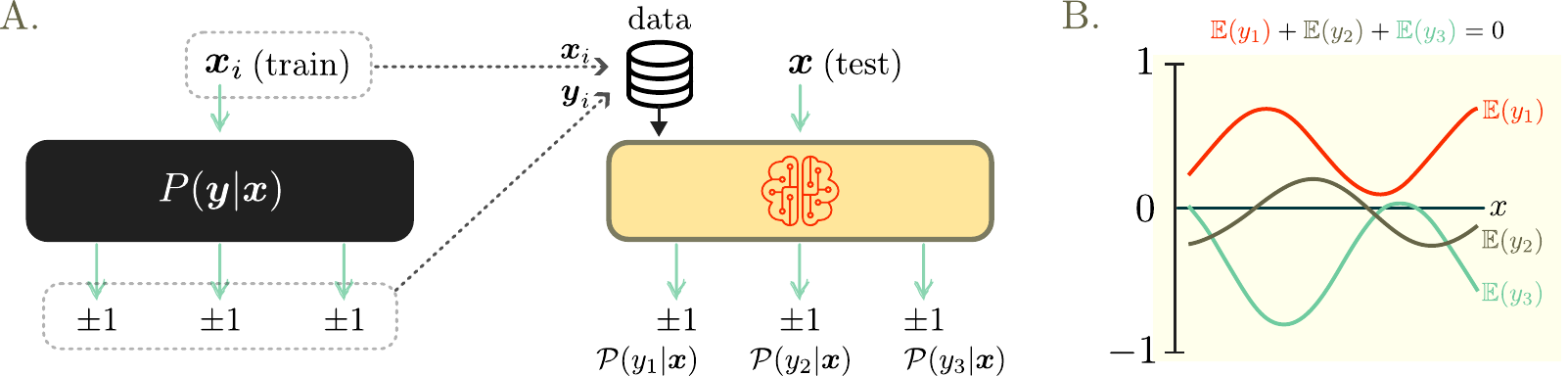}
    \caption{A. An example of the type of learning problem we consider in this work. Labels are generated for input training data $\vec{x}_i$ via a conditional process $P(\vec{y}_i\vert\vec{x}_i)$. Here, the labels take the form $\vec{y}_i=(y_i^{(1)},y_i^{(2)},y_i^{(3)})=(\pm1,\pm1,\pm1)$. The learning problem is to infer three probabilistic models $\mathcal{P}_1(y^{(1)}\vert\vec{x})$, $\mathcal{P}_2(y^{(2)}\vert\vec{x})$, $\mathcal{P}_3(y^{(3)}\vert\vec{x})$ that sample the individual labels for unseen input data. B. The data is assumed to satisfy a particular bias, which can be seen as a linear conservation law on the label space. Here, the sum of the expectation values of the labels is equal to zero for all $x$. We show that if a model encodes this as an inductive bias and is noncontextual, this implies constraints on the distributions $\mathcal{P}_k$, amounting to a limit on expressivity of model classes that are restricted to only noncontextual learning models.}
    \label{fig:mainfig}
\end{figure}

We adopt a notion of contextuality called \emph{generalised contextuality} \cite{gencon}, introduced by Spekkens in 2004. 
Loosely speaking, it refers to the fact that (i) there are different experimental procedures (called contexts) in the theory that are indistinguishable\footnote{For example, one can prepare a maximally mixed stated of a qubit by mixing states in the $z$ or $x$ basis; both procedures result in the same density matrix and are therefore indistinguishable.}, and (ii) any statistical model that reproduces the predictions of the theory must take these contexts into account. With this choice, our first task will then be to introduce a framework to talk about generalised contextuality in machine learning (Section \ref{sec:contextuality}). This was missing in previous works \cite{conml1,conml2}, which prove consequences for learning based on phenomena originating from contextuality, but do not attempt to define a notion of contextuality for machine learning that captures a wide range of models. Our general philosophy will be that the framework should depend purely on what a learning model can do, and not on the details of how it does it; i.e., the framework should be independent of the theory on which the models are built. This is necessary to have a framework that treats quantum and classical algorithms on the same footing, and ultimately involves adopting definitions in a similar spirit to the notion of operational contextuality as recently described in \cite{mischa}. 

We mostly focus on a paradigm of machine learning called multi-task learning \cite{mt0,mt1}, in which the aim is to simultaneously learn a number of separate models for a collection of different (but typically correlated) tasks. Multi-task learning scenarios are conceptually similar to commonly studied contextuality scenarios, and this similarity leads us to a definition of what it means for a multi-task model to be contextual (Section \ref{sec:contextml}). Although the focus on multi-class learning problems appears restrictive, as the separation between tasks is arbitrary at the mathematical level, we also arrive at a notion of contextuality in the single task setting (Section \ref{sec:relative}). In particular, we argue that it makes sense to think of contextuality as a property relative to a particular inductive bias of a model, rather than a property of the model as a whole.

Once we have described our framework, our second task will be to identify specific learning problems for which contextuality plays a role (Section \ref{sec:limits}). We show that this is the case when learning probabilistic models from data sets which feature a linearly conserved quantity in a discrete label space (see Figure \ref{fig:mainfig}). Such data sets can arise naturally from experiments involving conserved quantities, zero-sum game scenarios \cite{zs1,zs2}, logistics with conserved resources, substance diffusion \cite{substance1,substance2,substance3} in biological systems, and human mobility \cite{alessandretti2018evidence} and migration \cite{substancemigration}. We show that the ability of a model to encode the conserved quantity as an inductive bias directly links to a central concept in generalised contextuality, called \emph{operational equivalence}. This results in a constraint on noncontextual learning models that encode the desired bias, which amounts to a limit on the expressivity of noncontextual model classes. For certain data sets, this limitation can negatively impact generalisation performance due to the lack of a suitable model within the class that matches the underlying data distribution; in such cases contextuality may therefore be required for learning. To illustrate this point, in Section \ref{sec:rps} we construct a toy problem based on the rock, paper, scissors zero-sum game and prove precise limits on the expressivity of noncontextual model classes that attempt to learn the payoff behaviour of the game. 

In the final part of the work, we study the performance of quantum models for problems that involve our contextuality-inspired bias (Section \ref{sec:encodingbias}). We first describe two approaches to construct quantum ans\"{a}tze encoding the bias. The first of these encodes the bias into the state structure of the ansatz, and exploits tools from geometric quantum machine learning \cite{geoml1,geoml2,geoml3}. The second approach encodes the bias into the measurement structure, and we present a new family of measurements to this end that may be of independent interest. We then use these tools in a simple numerical investigation (Section \ref{sec:surrogate}), inspired by a recent work of Schreiber et al.\ \cite{surrogates}. Using the fact that quantum machine learning models are equivalent to truncated Fourier series \cite{gil2020input,schuld2021effect}, the authors of \cite{surrogates} define the notion of a classical surrogate model: a linear Fourier features model that has access to the same frequencies of the quantum model, but which lacks its specific inductive bias. The authors found that classical surrogate model classes perform better than quantum model classes on a wide range of regression tasks, the message being that it is still unclear what the inductive bias of quantum machine learning is useful for. In our numerical study, we show that a quantum model class that encodes our contextuality-inspired bias achieves a lower generalisation error than the corresponding surrogate model classes at a specific learning task, even after allowing for regularisation in the surrogate model. We argue that this is due to the fact that the bias cannot be easily encoded into the surrogate model class, which therefore cannot exploit this information during learning. 

In Section \ref{sec:applications} we elaborate on a number of areas where  contextuality-inspired inductive bias can be expected to play a role in learning. Many of these areas are classical in nature, and therefore suggests that quantum machine learning may be suited to tackling classical learning problems with a specific structure. Finally, in Section \ref{sec:outlook}, we outline our vision for this line of research and the possible next steps to take. Overall, we hope our approach and framework will lead to a new way of thinking about quantum machine learning, and ultimately lead to the identification of impactful problems where the specific structure of quantum theory makes quantum models the machine learning models of choice.

\tableofcontents

\section{Generalised contextuality}\label{sec:contextuality}
In this section we lay the theoretical groundwork that will lead us to a definition of contextuality for learning models. Some readers may wish to read section \ref{sec:main} first, which presents an overview of our main theoretical results, and can be understood without precise knowledge of the definitions presented here or in Section \ref{sec:contextml}. 

To define contextuality in machine learning we must first decide on a framework of contextuality, since there exist inequivalent frameworks in the literature. We have chosen to adopt that of \emph{generalised contextuality} \cite{gencon}, which from hereon we will often often refer to as simply `contextuality'. This is a modern version of contextuality that is commonly adopted in current foundations research \cite{conref0,conref1,conref2,conref3}. Like the more orthodox notion of Kochen-Specker contextuality \cite{kscon}, generalised contextuality has also been connected to speedup in quantum computation \cite{schmid2022uniqueness,needcon}, where it has been shown to be necessary for computational advantage. It is also based on a single, simple principle (described in section \ref{sec:gencon}) that subsumes older notions of contextuality such as Kochen-Specker contextuality \cite{kscon}, as well as Bell nonlocality \cite{wright2022invertible}. For this reason it is particularly attractive from a theoretical perspective.

Following our philosophy that contextuality for machine learning should reflect what the learning model does, rather than how it does it\footnote{For example, our notion of contextuality should not differentiate between a quantum model and a classical simulation of the same model, since the two are computationally equivalent from the perspective of the user. This perspective is normal in computation; for example, the concept of run-time is ignorant of the underlying theory.}, this will lead us to a definition that is inline with the concept of operational contextuality that was recently defined by Gitton and Woods \cite{mischa}. For this reason we will follow quite closely the notation and language of \cite{mischa}, and we encourage motivated readers to read the associated sections of that article, which give a more detailed introduction to generalised contextuality than we attempt here. In order to tailor the discussion to the machine learning audience, we introduce contextuality based on a user interacting with a machine learning model, however we remark that the framework of contextuality extends beyond this to general experimental scenarios.

\subsection{Procedures, preparations and effects}
The simplest and most commonly studied scenario of contextuality is the prepare-and-measure scenario (see Figure \ref{fig:models}A). For our purposes, we imagine a user with access to a machine learning model, who can prepare the model in a specific state by inputting data into it. 
In contextuality, this is called a \emph{preparation}, denoted $\Sb\in\texttt{Preps}$ (the letter $\Sb$ is used since a preparation can be thought of as a state), and is described via a \emph{procedure}, denoted $L$. A procedure is simply a list of actions describing in words what the user does, e.g., `input the data point $\vec{x}$ into the model'. We also assume that the user can query the model for different tasks and observe the corresponding predictions. These events are described by \emph{effects} in contextuality. An effect, denoted $\Eb = \omega\vert L\in \texttt{Effects}$ is a procedure $L$ together with a specific observation $\omega\in W_L$ that signals success, where $W_L$ is the set of all observations that may follow $L$. For example `query the model for task 1' (procedure) and `observe the label +1' (observation). An effect can therefore be thought of as a binary yes/no question related to a specific procedure\footnote{Whereas preparations can be thought of as the analogue of quantum states, effects can be thought of as the analogue of measurement operators.}.  

Note that at this level of description we have not specified anything about how the model works (it could be quantum or classical for example), and are concerned only with descriptions of the actions and observations of the user that interacts with it. For this reason, this approach is often called an \emph{operational} approach; hence the common appearance of the term `operational' in the proceeding sections.  

\begin{figure}
    \centering
    \includegraphics[width=\textwidth]{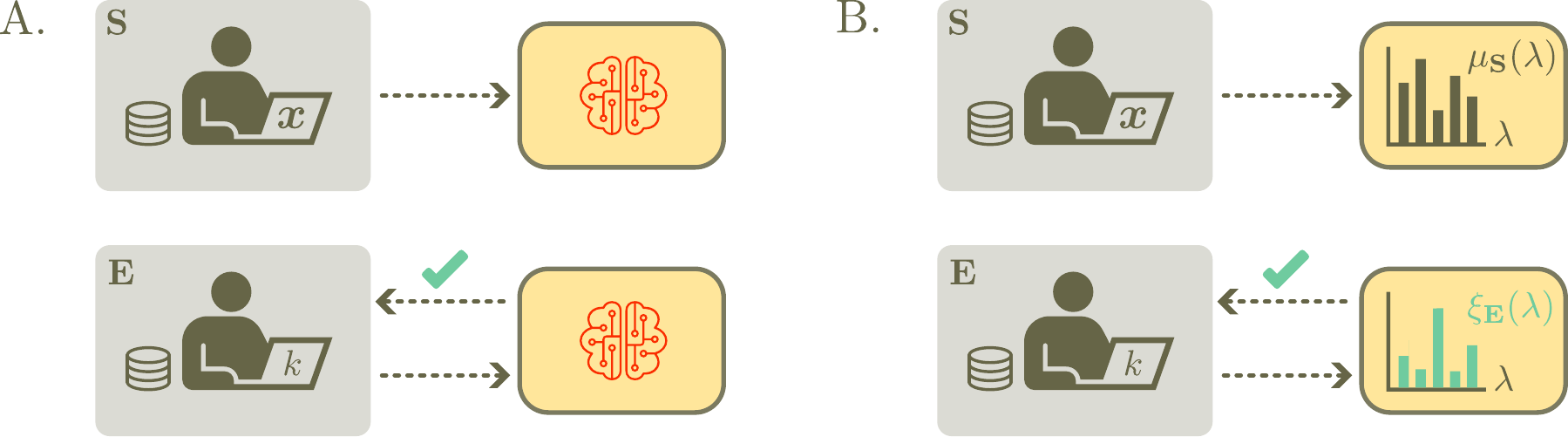}
    \caption{A. (top) A prepare-and-measure scenario. A preparation $\Sb$ is a procedure (a list of actions) that a user carries out. In this example, the procedure is to input the data $\vec{x}$ into the machine learning model. (bottom) An effect $\Eb$ is another procedure that additionally has an observation that signals success of the procedure. Here, the procedure is to query the model with some task (labelled $k$) and success corresponds to a particular label being returned. B. An ontological model is a particular causal model to explain the statistics of a prepare-and-measure scenario. The preparation results in some probabilistic mixture $\mu_{\Sb}(\lambda)$ over a set of ontic states $\lambda$. The effect is a function $\xi_{\Eb}(\lambda)$ that returns the probability of the successful outcome given the ontic state $\lambda$. For us, $\lambda$ can be thought of a specific memory state of the learning model, and $\xi_{\Eb}(\lambda)$ a function that gives the probability of the learning model returning a specific label for some task.}
    \label{fig:models}
\end{figure}

\subsection{Operational statistics}
For every $\Sb\in\texttt{Preps}$ and every $\Eb\in\texttt{Effects}$, we assume there is a corresponding probability $P(\Eb\vert \Sb)$ that describes the likelihood of the successful outcome of $\Eb$ given $\Sb$, and we call the set of all $P(\Eb\vert \Sb)$ the \emph{operational statistics}. In our case, $P(\Eb\vert \Sb)$ is simply the probability to observe a particular label for some task (given by $\Eb$) conditioned on certain data being input into the model (given by $\Sb)$. In addition to the effects that correspond to the possible predictions of the model, we also include the trivial effect $\Omega$ that always occurs for any preparation,
\begin{align}\label{trivial}
  P(\Omega \vert \Sb)=1 \quad \forall \, \Sb \in\texttt{Preps},
\end{align}
and the null effect $\emptyset$ that never occurs,
\begin{align}\label{null}
  P(\emptyset \vert \Sb)=0 \quad \forall \, \Sb \in\texttt{Preps}.
\end{align}
Operationally speaking, these effects correspond to the user simply ignoring the model and declaring either a successful or failed observation. We naturally assume that probabilities of effects corresponding to two mutually incompatible observations are additive. For example, if $\Eb_{0 \lor 1}=0\lor 1\vert L$ is an effect corresponding to observing a label prediction $0$ or $1$ and $\Eb_{0}=0\vert L$ and $\Eb_{1}=1\vert L$ are those for labels $0$ and $1$ independently, then 
\begin{align}
    P(\Eb_{0 \lor 1}\vert \Sb) = P(\Eb_{0}\vert \Sb)+P(\Eb_{1}\vert \Sb).
\end{align}
With this, we can now define a measurement 
\begin{align}
    \Mb = \{\Eb_1, \Eb_2, \cdots ,\Eb_d\}
\end{align}
to be a collection of effects corresponding to the same procedure $L_{\Mb}$ with mutually incompatible observations $\omega_i$ whose probabilities sum to 1:
\begin{align}
    \Eb_i = \omega_i\vert L_{\Mb}, \quad \sum_i P(\Eb_i\vert \Sb)=1\quad \forall\, \Sb \in\texttt{Preps}.
\end{align}
we call the triple
    $\{\texttt{Preps},\texttt{Effects},P(\Eb\vert\Sb)\}$
of possible preparations, effects and operational statistics, an \emph{operational scenario}\footnote{This is analogous to the concept of `operational theory' encountered in the generalized contextuality literature; here we use the word `scenario' instead since we are not concerned with an entire physical theory.}.

\subsection{Convex mixtures}
A key ingredient of generalised contextuality is the ability to perform convex mixtures; i.e., the user is assumed to have access to a trusted source of randomness that they can use to probabilistically mix procedures. For example, given two preparations $\Sb_1$ and $\Sb_2$ with procedures $L_1$ and $L_2$, the user may flip a biased coin (with bias $p$) and perform either $L_1$ or $L_2$. We call the resulting preparation a \emph{preparation density}, written
\begin{align}
    s = p \Sb_1 + (1-p) \Sb_2,
\end{align}
which results in a corresponding mixture of operational statistics, 
\begin{align}\label{add1}
P(\Eb\vert p \Sb_1 + (1-p) \Sb_2) = p P(\Eb\vert \Sb_1) + (1-p)P(\Eb\vert \Sb_2) \quad \forall \Eb\in\texttt{Effects}.
\end{align}
Similarly given two effects $\Eb_{1}=\omega_1\vert L_1$ and $\Eb_{2}=\omega_2 \vert L_2$, the user can flip a biased coin and realise the \emph{effect density}
\begin{align}
    e = p \Eb_1 + (1-p) \Eb_2,
\end{align}
such that
\begin{align}\label{add2}
P(p \Eb_1 + (1-p) \Eb_2\vert \Sb) = p P(\Eb_1\vert \Sb) + (1-p)P(\Eb_2\vert \Sb) \quad \forall \Sb \in\texttt{Preps}.
\end{align}
By allowing arbitrary probabilistic mixtures of preparations and effects, we can therefore prepare any preparation density $s\in\text{conv}(\texttt{Preps})$ and any effect density $e\in\text{conv}(\texttt{Effects})$, where $\text{conv}$ denotes the convex hull of a set.

\subsection{Operational equivalences}\label{sec:opeq}
At the heart of generalised contextuality is the notion of operational equivalence, which relates to the possibility of distinguishing between pairs of preparation or effect densities. We say that two preparation densities $s_1,s_2$ are \emph{operationally equivalent}, denoted $\sim$, iff they give identical predictions for all effects,
\begin{align}
    s_1 \sim s_2 \iff P(\Eb \vert s_1) =  P(\Eb \vert s_2) \quad \forall\,\Eb\in\texttt{Effects}. 
\end{align}
Similarly, two effect densities $e_1,e_2$ are operationally equivalent iff they have the same probability of success for all preparations,
\begin{align}
    e_1\sim e_2 \iff P(e_1 \vert \Sb) = P(e_2 \vert \Sb) \quad \forall \, \Sb\in\texttt{Preps}.
\end{align}
Thus, a pair of preparations or effects are operationally equivalent if and only if they are indistinguishable with respect to any of the possible procedures and observations in the operational scenario. 

%

\subsection{Ontological models}\label{sec:ontmodel}
An ontological model is a specific causal model (see Figure \ref{fig:models}B) to describe the operational statistics of an operational scenario. In an ontological model, each preparation $\Sb$ is assigned a probability density function $\mu_{\Sb}(\lambda)$ called an \emph{ontic state distribution} over a set of so-called \emph{ontic states}\footnote{One could also call these hidden variables.} $\lambda \in \Lambda$. The idea here is that the set $\Lambda$ of ontic states represents the set of all possible accessible physical states in some underlying theory. Returning to our machine learning motivation, we can think of $\lambda$ as being a specific memory state of the learning model and $\Lambda$ the set of all such memory states. For example, if the model is a neural network then $\lambda$ denotes the physical memory state of the computer on which it is stored. $\mu_{\Sb}(\lambda)$ is therefore the probability distribution over memory states induced by some data preparation $\Sb$. 

The probability of an effect given some preparation is causally mediated through the variable $\lambda$. 
Each effect $\Eb$ is assigned a function $\xi_{\Eb}(\lambda):\Lambda\rightarrow [0,1]$ called an \emph{ontic response function} that returns the probability of $\Eb$ given $\lambda$. In a machine learning setting $\xi_{\Eb}(\lambda)$ typically corresponds to the probability of returning a particular label for some task given the memory state $\lambda$. From the law of total probability we therefore have
\begin{align}\label{ontmodel}
    P(\Eb\vert \Sb) = \int_\Lambda\text{d}\lambda\, \mu_{\Sb}(\lambda)\xi_{\Eb}(\lambda).
\end{align}
Ontic state distributions and ontic response functions for preparation and effect densities are defined using the same additivity relations as \eqref{add1}, \eqref{add2},
\begin{align}
    \mu_{p\Sb_1+(1-p)\Sb_2}(\lambda) = p\mu_{\Sb_1}(\lambda) + (1-p)\mu_{\Sb_2}(\lambda), \quad
    \xi_{p\Eb_1+(1-p)\Eb_2}(\lambda) = p\xi_{\Eb_1}(\lambda) + (1-p)\xi_{\Eb_2}(\lambda).
\end{align}
An \emph{ontological model} of an operational scenario is defined by a triple $\{\Lambda, \{\mu_{\Sb}(\lambda)\}, \{\xi_{\Eb}(\lambda)\}\}$ of ontic states, ontic state distributions and ontic response functions that reproduce the operational statistics via \eqref{ontmodel}. 

\subsection{Generalised noncontextuality of ontological models}\label{sec:gencon}
Generalised noncontextuality is a constraint, called \emph{ontological equivalence}, that is imposed on ontological models. Within a given ontological model of some operational scenario, we say that two preparation densities $s_1$ and $s_2$ are ontologically equivalent if
\begin{align}
\mu_{s_1}(\lambda)  = \mu_{s_2}(\lambda) \quad \forall \,\lambda,
\end{align}
i.e., they result in the same probability density over ontic states. Similarly, two effect densities $e_1$ and $e_2$ are ontologically equivalent if they have the same ontic response function:
\begin{align}
    \xi_{e_1}(\lambda)  = \xi_{e_2}(\lambda) \quad \forall \,\lambda.
\end{align}
Generalised noncontextuality can now be understood as the following principle, which is often motivated by Leibniz's principle of the identity of indiscernibles \cite{leibniz}:
\begin{center}
    operational equivalence $\iff$ ontological
    equivalence,
\end{center}
or equivalently
\begin{align}
    s_{1}\sim s_{2} &\iff   \mu_{s_1}(\lambda)  = \mu_{s_2}(\lambda) \quad \forall \,\lambda, \label{onteq1} \\
    e_{1}\sim e_{2} &\iff \xi_{e_1}(\lambda)  = \xi_{e_2}(\lambda) \quad \forall \,\lambda. \label{onteq2} 
    \end{align}
In other words, \emph{if two preparation of effect densities are indistinguishable, then they should be described in the same way in the ontological model}. At this point it can be useful to make a link to Kochen-Specker noncontextuality for those that are familiar. In Kochen-Specker noncontextuality, one assigns the same deterministic response function to projectors that appear in different measurement contexts. This follows from an application of \eqref{onteq1} and \eqref{onteq2}, where the projector that appears in two different measurements is understood as two distinct but operationally equivalent effects, and the determinism of the response functions is a consequence of preparation noncontextuality \cite{kunjwal2015kochen}. As we will see however, there are other operational equivalences that are not of this form. As a result, the framework is applicable to scenarios beyond the Kochen-Specker setting; hence the name \emph{generalised} contextuality.

A noncontextual ontological model of an operational scenario can now be defined as follows.
\begin{defo}[Noncontextual ontological model]{def:ncommain}
A noncontextual ontological model of an operational scenario $ \{\texttt{Preps},\texttt{Effects},P(\Eb\vert\Sb)\}$ is a triple \begin{align}
\{\Lambda, \mu_{\Sb}(\lambda), \xi_{\Eb}(\lambda)\}
\end{align}
of ontic states, ontic state distributions and ontic response functions satisfying the constraints \eqref{onteq1} and \eqref{onteq2} such that 
\begin{align}
    P(\Eb\vert \Sb) = \int_\Lambda\text{d}\lambda\, \mu_{\Sb}(\lambda)\xi_{\Eb}(\lambda) \quad \forall \Sb\in\texttt{Preps},\; \forall\Eb\in\texttt{Effects}.
\end{align}
\end{defo}
In prepare-and-measure scenarios where operational equivalences exist, the constraints \eqref{onteq1} and \eqref{onteq2} limit the set of operational statistics that admit a noncontextual ontological model to a subset of all possible statistics. If the operational statistics lie outside of this set, they are therefore called \emph{contextual}. 

\section{Contextuality of multi-task machine learning models}\label{sec:contextml}
We now have all the necessary ingredients we need to define a notion of contextuality for machine learning. As in the previous section, we consider a scenario that involves a user with access to a learning model, who can input data $\vec{x}$ into the model and observe predictions $y$ for $m$ different tasks. A learning model in this scenario is therefore a device that samples from one of the $m$ conditional probability distributions $\mathcal{P}_{\theta}^1(y\vert\vec{x}),\mathcal{P}_{\theta}^2(y\vert\vec{x}),\cdots, \mathcal{P}_{\theta}^m(y\vert\vec{x})$, and we denote the model by
\begin{align}\label{mtmodel}
h(\theta) =  \{\mathcal{P}_{\theta}^1(y\vert\vec{x}),\mathcal{P}_{\theta}^2(y\vert\vec{x}),\cdots, \mathcal{P}_{\theta}^m(y\vert\vec{x})\}.
\end{align}
Here, $\theta$ denotes some specific fixed parameters of the associated model class $H$,
\begin{align}\label{mtclass}
    H=\{h(\theta)\}_\theta = \{\{\mathcal{P}_\theta^1(y\vert\vec{x}),\mathcal{P}_\theta^2(y\vert\vec{x}),\cdots,\mathcal{P}_\theta^m(y\vert\vec{x})\}\}_\theta,
\end{align} 
which are chosen according to some data-dependent figure-of-merit, such as empirical risk minimisation. 

Since the model $h(\theta)$ is capable of inference for a number of different tasks, we call it a \emph{multi-task model}, and $H$ a \emph{multi-task model class}. Multi-task models are an active area of machine learning research \cite{mt0,mt1}, and in well-aligned tasks, can lead to improved performance and lower data size requirements than an approach in which models are learned independently. In this section we focus on multi-task models to define our notion of contextuality. The reason for this is that this paradigm most naturally fits the existing contextuality literature, where the relevant scenarios involve a choice of measurements that we map to separate tasks. However, since what defines a task is rather ambiguous (i.e.~a group of tasks could be thought of as a single task), this leads to a notion of contextuality of learning models in the single task setting as well. Since this is quite a subtle point that will benefit from understanding the multi-task case, we defer the discussion and the corresponding definition for single-task models until Section \ref{sec:relative}.

We now imagine a user interacting with a trained multi-task model $h(\theta)$. Our definition of contextuality will apply to $h(\theta)$ rather than be a property of the entire model class $H$. The reason for this is that we aim to connect contextuality to the behaviour of the model at inference time, which is the job of a single trained multi-task model, rather than the entire class\footnote{That is not to say that definitions of contextuality that take into account the full model class are not interesting, however the interpretation would be somewhat different.}. With this choice we may now define the corresponding operational scenario in which to study contextuality. Following the logic of Section \ref{sec:contextuality}, this
consists of all the actions and observations related to a user interacting with the model $h(\theta)$ at inference time. \vspace{5pt}

\begin{defo}[Operational scenario of a multi-task model]{def:prepsmeas}
An operational scenario of a multi-task model is a triple $ \{\texttt{Preps},\texttt{Effects},P(\Eb\vert\Sb)\}$. The set \texttt{Preps} consists of preparations $\Sb_{\vec{x}}$, where $\Sb_{\vec{x}}$ corresponds to a procedure whereby a user inputs data $\vec{x}$ into the memory of the model. The set \texttt{Effects} consists of the effects $\Eb^k_y$, $\Omega$, $\emptyset$, where $\Eb^k_y$ corresponds to a procedure whereby the user queries the model for task $k$ and observes the label $y$, and  $\Omega$ and $\emptyset$ are the trivial and null effects respectively. The operational statistics are given by the model predictions $P(\Eb^k_{y}\vert\Sb_{\vec{x}})=\mathcal{P}_{\theta}^k(y\vert\vec{x})$.
\end{defo}
This choice of operational scenario then leads to a natural definition of noncontextuality for multi-task models. \vspace{5pt}
\begin{defo}[Noncontextual multi-task model]{def:opcontmodel}
A multi-task model
\begin{align}
   h(\theta)= \{\mathcal{P}_{\theta}^1(y\vert\vec{x}),\mathcal{P}_{\theta}^2(y\vert\vec{x}),\cdots, \mathcal{P}_{\theta}^m(y\vert\vec{x})\}
\end{align}\vspace{5pt}
is noncontextual iff the operational statistics $P(\Eb\vert \Sb)$ of the associated operational scenario (given by Definition \ref{def:prepsmeas}) admits a noncontextual ontological model. If such an ontological model does not exist, the learning model is said to be contextual. 
\end{defo}
With this definition, we now define a noncontextual model class to be any class that contains only noncontextual multi-task models.\vspace{5pt}

\begin{defo}[Noncontextual multi-task model class]{def:ncclass}
A multi-task model class
\begin{align}
   H = \{\{\mathcal{P}_{\theta}^1(y\vert\vec{x}),\mathcal{P}_{\theta}^2(y\vert\vec{x}),\cdots, \mathcal{P}_{\theta}^m(y\vert\vec{x})\}\}_{\theta}
\end{align}\vspace{5pt}
is noncontextual if every multi-task model in the class is noncontextual in the sense of Definition \ref{def:opcontmodel}. 
\end{defo}
We remark that the theory-independent approach we have taken above is different to the common perspective of generalised contextuality, where the operational scenario corresponds to the entire underlying physical theory. In Appendix \ref{app:def} we expand on this point and explain why our approach is necessary to arrive at a reasonable definition of contextuality for machine learning.

\section{Inductive bias and limits of expressivity}\label{sec:limits}
In this section we present our main theoretical insights by connecting the concept of inductive bias in machine learning to generalised contextuality. We use upper case Roman characters to denote random variables and lower case to denote fixed values of that variable, i.e., \ $P(a\vert x)\equiv P(A=a\vert X=x)$. 
Superscript in parentheses refers to a particular element of a vector, i.e., $\vec{x}^{(k)}$ is the $k^{\text{th}}$ element of $\vec{x}$. 

\subsection{Summary of main result}\label{sec:main}
We first give a summary of our main result, leading to Theorem \ref{thm:main}. We focus on a class of learning problems with the following form (see also Figure \ref{fig:mainfig}).

\begin{probo}[Statistical learning problems with linear bias]{def:genproblem}
Consider a data set 
\begin{align}
    \chi = \{\vec{x}_i, \vec{y}_i\}
\end{align}
of input data points $\vec{x}_i\in\mathbb{R}^d$ and labels $\vec{y}_i\in\{-1,1\}^{m}$ sampled from a data distribution $\mathcal{D}(\vec{x}, \vec{y})$ that satisfies a linear constraint
\begin{align}\label{lincon}
\mathbb{E}[Y^{(1)}\vert \vec{x}] + \mathbb{E}[Y^{(2)}\vert \vec{x}] + \cdots + \mathbb{E}[Y^{(m)}\vert \vec{x}] = 0 \quad \forall \vec{x},
\end{align}
where $Y^{(k)}$ is the random variable corresponding to the $k^{\text{th}}$ element of $\vec{y}$. Construct a multi-task model 
\begin{align}
     h(\theta) = \{\mathcal{P}_{\theta}^1(y\vert\vec{x}),\mathcal{P}_{\theta}^2(y\vert\vec{x}),\cdots, \mathcal{P}_{\theta}^m(y\vert\vec{x})\}
\end{align}
that samples from the $m$ marginal conditional distributions $\mathcal{D}(y^{(k)}\vert \vec{x})$ $(k=1,\cdots,m)$ for unseen instances of $\vec{x}$.
\end{probo}
%

\begin{figure}
    \centering
    \includegraphics[scale=0.7]{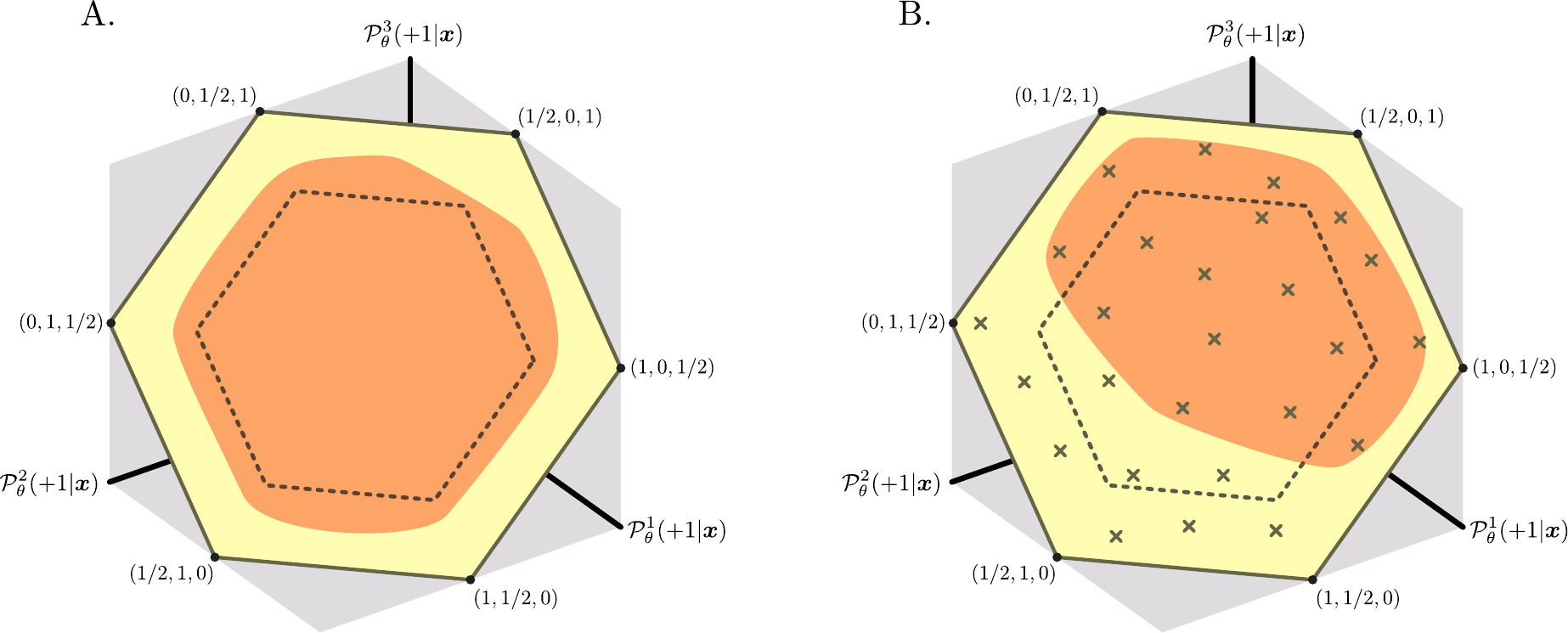}
    \caption{A. If a multi-task model encodes the bias \eqref{bias}, the behaviours $\vec{v}_{\vec{x}} = (\mathcal{P}_\theta^1(+1\vert \vec{x}),\mathcal{P}_\theta^2(+1\vert \vec{x}),\mathcal{P}_\theta^3(+1\vert \vec{x}))$ must lie in the yellow hexagon formed by the convex hull of six extremal points. The orange space denotes the convex hull of $V$, the set of possible behaviours for some specific multi-task model. If this space lies outside of the dashed inner hexagon (described in Theorem \ref{thm:main}), the multi-task model is necessarily contextual in the sense of Definition \ref{def:opcontmodel}. B. The crosses mark the ground truth behaviours of inputs from a test data set. Since the convex hull of the set of behaviours of a noncontextual model  (orange space) cannot contain the inner hexagon, any such model will be unable to sample the correct label distribution for all inputs in the test set. In this case, the noncontextual model will fail to approximately sample the correct label distributions for those test points with behaviours far from the orange space.}
    \label{fig:thmfig}
\end{figure}

The constraint \eqref{lincon} can be seen as a statistical conservation law on the label space: for any $\vec{x}$ the sum of their expectation values is equal to a constant. As we will see in section \ref{sec:linking}, noncontextual model classes that encode biases of this form have limited expressivity, which may impact learning from certain data sets. This is due to the inductive bias acting as an operational equivalence (Section \ref{sec:opeq}) in the contextuality scenario, which amounts to a constraint on every model in the class through the principle of generalised contextuality (Section \ref{sec:gencon}).

To illustrate how this works in practice, we study the particular case of Problem \ref{def:genproblem} for $m=3$. Consider a data set $\chi$ as described in problem \ref{def:genproblem}, and a multi-task model class $H$.
We assume that the model class encodes the bias \eqref{lincon} that is present in the data, so that \eqref{lincon} holds for each multi-task model $h(\theta)$ in $H$. Translating the expectation values $\mathbb{E}[Y^{(k)}\vert \vec{x}]= \mathcal{P}_\theta^k(+1\vert \vec{x})-\mathcal{P}_\theta^k(-1\vert \vec{x})$ into probabilities, this implies
%
\begin{align}
    \mathcal{P}_\theta^1(+1\vert \vec{x})+\mathcal{P}_\theta^2(+1\vert \vec{x})+\mathcal{P}_\theta^3(+1\vert \vec{x}) = \mathcal{P}_\theta^1(-1\vert \vec{x})+\mathcal{P}_\theta^2(-1\vert \vec{x})+\mathcal{P}_\theta^3(-1\vert \vec{x}) \quad \forall\,\theta\;\forall\,\vec{x},
\end{align}
or using $\mathcal{P}_\theta^k(+1\vert \vec{x})=1-\mathcal{P}_\theta^k(-1\vert \vec{x})$,
\begin{align}\label{bias}
    \mathcal{P}_\theta^1(+1\vert \vec{x})+\mathcal{P}_\theta^2(+1\vert \vec{x})+\mathcal{P}_\theta^3(+1\vert \vec{x}) = \frac{3}{2} \quad \forall \theta\;\forall\,\vec{x}.
\end{align}

In order to describe how the bias \eqref{bias} leads to constraints on noncontextual models, we now give a geometric interpretation to each multi-task model in the class. Consider a specific multi-task model (with fixed parameters $\theta$) and define the vector 
\begin{align}
\vec{v}_{\vec{x}} = (\mathcal{P}_\theta^1(+1\vert \vec{x}),\mathcal{P}_\theta^2(+1\vert \vec{x}),\mathcal{P}_\theta^3(+1\vert \vec{x})).
\end{align}
We call $\vec{v}_{\vec{x}}$ a \emph{behaviour} for input $\vec{x}$. Since $\mathcal{P}_\theta^k(-1\vert \vec{x})=1-\mathcal{P}_\theta^k(+1\vert \vec{x})$, the set of possible behaviours that can be produced by the model is
\begin{align}\label{vset}
    V = \{\vec{v}_{\vec{x}}\vert \vec{v}_{\vec{x}} \in D_{\vec{x}}\},
\end{align}
where $D_x$ is the domain of the input data. Since each probability in $\vec{v}_{\vec{x}}$ is bounded in $[0,1]$, the set $V$ generally forms a three-dimensional subspace of the  cube. The bias \eqref{bias} however is a linear constraint that restricts $V$ to live in a  to a two-dimensional hexagon that is the convex hull of the six extremal behaviours  (see Figure \ref{fig:thmfig}A)
\begin{align}\label{vprobs}
    \vec{v}_1 = (1,0,\half), \quad \vec{v}_2 = (0,1,\half), \quad  \vec{v}_3 = (\half,1,0), \quad \vec{v}_4 = (\half,0,1),\quad
    \vec{v}_5 = (0,\half,1), \quad \vec{v}_6 = (1,\half,0). \quad
\end{align}
In Figure \ref{fig:thmfig}A, the orange space shows the convex hull of $V$ for such a model. We now consider another space $V_{\eta}$ that is the convex hull of the six behaviours 
\begin{align}
    \vec{u}_i(\eta) = \eta\vec{v}_i + (1-\eta)(\half,\half,\half) \quad i=1,\cdots,6
\end{align}
where $0<\eta<1$. This space also forms a (smaller) hexagon, which for $\eta=\frac{2}{3}$ is shown by the dashed lines in Figure \ref{fig:thmfig}A. As we will show shortly, if the set $V_{2/3}$ is strictly contained in the convex hull of $V$, then the corresponding multi-task model is necessarily contextual. This is made precise by the following theorem. \vspace{5pt}

\begin{theo}[Necessity of contextuality for biased multi-task models]{thm:main}
Consider a multi-task model $
   \{\mathcal{P}_{\theta}^1(y\vert\vec{x}),\mathcal{P}_{\theta}^2(y\vert\vec{x}),\mathcal{P}_{\theta}^3(y\vert\vec{x})\}
  $
that encodes the bias \eqref{bias}.\vspace{5pt}
If the set $V_{2/3}$ given by the convex hull of points
\begin{align}\label{points}
    \vec{u}_i(\eta) = \eta\vec{v}_i + (1-\eta)(\half,\half,\half) \quad i=1,\cdots,6
\end{align}
for $\eta=\frac{2}{3}$ is strictly contained the convex hull of 
\begin{align}
V = \{\vec{v}_{\vec{x}}\vert \vec{x} \in D_{\vec{x}}\} =  \{(\mathcal{P}_{\theta}^1(+1\vert \vec{x}),\mathcal{P}_{\theta}^2(+1\vert \vec{x}),\mathcal{P}_{\theta}^3(+1\vert \vec{x}))\vert \vec{x} \in D_{\vec{x}}\}
\end{align}
then the multi-task model is contextual in the sense of Definition \ref{def:opcontmodel}. 
\end{theo}

Phrased in the contra-positive, this implies that if a multi-task model encodes the bias \eqref{bias} and is noncontextual, the convex hull of possible behaviours cannot contain $V_{2/3}$. This amounts to a limit on the expressivity of noncontextual multi-task model classes, since each model in the class is subject to this restriction. For certain data sets this may impact the generalisation ability of the model class (see Figure \ref{fig:thmfig}B). For example, consider a test data set that contains a significant fraction of inputs $\vec{x}_i$ whose ground truth behaviours are sufficiently close to the six points $\vec{v}_i$. Then since the set of behaviours of any noncontextual model cannot contain $V_{2/3}$, all models in the class will struggle to approximately sample the correct label distributions on some of the inputs, which may hinder generalisation ability.

\subsection{Linking inductive bias to operational equivalence}\label{sec:linking}
In this subsection and the next we prove Theorem \ref{thm:main}. To do this we will connect the inductive bias of the model to the notion of operational equivalence of generalised contextuality. Our first task is to identify the relevant operational scenario. Following Definition \ref{def:prepsmeas}, the set $\texttt{Preps}$ contains the preparations $\Sb_{\vec{x}}$ that correspond to inputting a data point $\vec{x}\in D_x$ into the model. The set $\texttt{Effects}$ contains the trivial effect $\Omega$, null effect $\emptyset$, and the six effects $\Eb_{\pm}^1$, $\Eb_{\pm}^2$, $\Eb_{\pm}^3$, which correspond to observing the label $y=\pm1$ for the three tasks. The operational statistics are simply the model predictions
\begin{align}\label{opstats2}
    P(\Eb^k_y \vert\Sb_{\vec{x}}) = \mathcal{P}_{\theta}^k(y\vert\vec{x}) \quad k=1,2,3.
\end{align}

An operational equivalence connected to the bias \eqref{bias} can be found as follows. Imagine we choose one of the tasks uniformly at random and observe the label of that task. In the operational scenario this is described by a measurement
\begin{align}
   \Mb = \{e_+,e_-\}= \{ \frac{1}{3}\Eb^1_{+}+\frac{1}{3}\Eb^2_{+}+\frac{1}{3}\Eb^3_{+}, \frac{1}{3}\Eb^1_{-}+\frac{1}{3}\Eb^2_{-}+\frac{1}{3}\Eb^3_{-} \} .
\end{align}
Let us look at the probability of obtaining the first outcome of this measurement when performed on some preparation $\Sb_{\vec{x}}$. Since the bias is satisfied, it follows from \eqref{bias} that 
\begin{align}
    P(e_+\vert\Sb_{\vec{x}})&=\frac{1}{3}(P(\Eb^1_{+}\vert\Sb_{\vec{x}})+P(\Eb^2_{+}\vert\Sb_{\vec{x}})+P(\Eb^3_{+}\vert\Sb_{\vec{x}})) \nonumber \\
    &=  \frac{1}{3}(\mathcal{P}_\theta^1(+1\vert \vec{x})+\mathcal{P}_\theta^2(+1\vert \vec{x})+\mathcal{P}_\theta^3(+1\vert \vec{x})) = \frac{1}{2} \quad \forall\;\Sb_{\vec{x}}.
\end{align}
Similarly, the probability of the second outcome is 
\begin{align}
     P(e_-\vert\Sb_{\vec{x}})&=\frac{1}{3}(P(\Eb^1_{-}\vert\Sb_{\vec{x}})+P(\Eb^2_{-}\vert\Sb_{\vec{x}})+P(\Eb^3_{-}\vert\Sb_{\vec{x}})) \nonumber \\ &=  \frac{1}{3}(\mathcal{P}_\theta^1(-1\vert \vec{x})+\mathcal{P}_\theta^2(-1\vert \vec{x})+\mathcal{P}_\theta^3(-1\vert \vec{x})) \nonumber\\
    &= 1- \frac{1}{3}(\mathcal{P}_\theta^1(+1\vert \vec{x})+\mathcal{P}_\theta^2(+1\vert \vec{x})+\mathcal{P}_\theta^3(+1\vert \vec{x})) = \frac{1}{2} \quad \forall\;\Sb_{\vec{x}}.
\end{align}
This means that if we chose a task at random, we are equally likely to observe a positive or negative label. The two effect densities $e_+$ and $e_-$ in this measurement are therefore operationally equivalent\footnote{One could use the operational equivalence $\frac{1}{3}\Eb^1_{+}+\frac{1}{3}\Eb^2_{+}+\frac{1}{3}\Eb^3_{+} \sim \frac{1}{2}\Omega + \frac{1}{2}\emptyset$ here instead to the same effect.},
\begin{align}\label{opeqv_meas}
    \frac{1}{3}\Eb^1_{+}+\frac{1}{3}\Eb^2_{+}+\frac{1}{3}\Eb^3_{+} \sim \frac{1}{3}\Eb^1_{-}+\frac{1}{3}\Eb^2_{-}+\frac{1}{3}\Eb^3_{-}.
\end{align}

An operational equivalence related to the preparations can be found as follows. Since the set $V_{2/3}$ is strictly contained in the convex hull of $V$, it follows that the convex hull of $V$ contains the six points $\vec{u}_i(\eta)$ for some $\eta>\frac{2}{3}$. We therefore have 
\begin{align}\label{stime}
    \vec{u}_i(\eta)=(P(\Eb_+^1\vert s_i),P(\Eb_+^2\vert s_i),P(\Eb_+^3\vert s_i)) \quad i=1,\cdots,6
\end{align}
for some preparation densities $s_1,\cdots,s_6$. Note that we also have 
\begin{align}\label{uvecs}
    \frac{1}{2}( \vec{u}_1(\eta) +  \vec{u}_2(\eta))=  \frac{1}{2}( \vec{u}_3(\eta) +  \vec{u}_4(\eta)) =  \frac{1}{2}( \vec{u}_5(\eta) +  \vec{u}_6(\eta))=(\frac{1}{2},\frac{1}{2},\frac{1}{2}).
\end{align}
 It follows from  \eqref{add2} that
 \begin{align}
(P(\Eb_+^1\vert \half s_i+\half s_{i+1}),P(\Eb_+^2\vert \half s_i+\half s_{i+1}),P(\Eb_+^3\vert \half s_i+\half s_{i+1})) = (\frac{1}{2},\frac{1}{2},\frac{1}{2})
 \end{align}
 for $i=1,3,5$. From this we can conclude that equally mixing any of the pairs of preparation densities $(s_1,s_2)$, $(s_3,s_4)$, or $(s_5,s_6)$ results in another preparation density that returns uniformly random labels for the three tasks. These mixtures are therefore indistinguishable and we have the operational equivalences
\begin{align}\label{opeqv_prep}
    \frac{1}{2}s_1+ \frac{1}{2}s_2\sim\frac{1}{2}s_3+ \frac{1}{2}s_4\sim\frac{1}{2}s_5+ \frac{1}{2}s_6. 
\end{align}

\subsection{Bounding noncontextual learning models}
Assuming noncontextuality, we now derive constraints on the operational statistics via the operational equivalences \eqref{opeqv_meas} and \eqref{opeqv_prep}, leading to the result of Theorem \ref{thm:main}. To do this, we use the following lemma. \vspace{5pt}

\begin{lem}[Noncontextuality inequality]{lemma}
Consider an operational scenario of a multi-task model for Problem \ref{def:genproblem} with $m=3$, that satisfies the operational equivalences \eqref{opeqv_meas} and \eqref{opeqv_prep}. If the multi-task model is noncontextual, then the operational statistics satisfy the inequality
\begin{align}\label{ineq}
P(\Eb_+^1\vert s_1)+P(\Eb_+^2\vert s_3)+P(\Eb_+^3\vert s_5)\leq \frac{5}{2},
\end{align}
\end{lem}
\begin{proof}
Consider a noncontextual ontological model of the operational scenario. This consists of ontic state distributions $\mu_s(\lambda)$ for every preparation density $s$ and ontic response functions $\xi_{e}(\lambda)$ for every effect density $e$. The effect densities are convex mixtures of the eight effects  $\Eb^1_{+},\Eb^2_{+},\Eb^3_{+},\Eb^1_{-},\Eb^2_{-},\Eb^3_{-},\Omega,\emptyset$. We define the vectors 
\begin{align}\label{vecxi}
    \vec{\xi}(\lambda) = (\xi_{\Eb^1_+}(\lambda),\xi_{\Eb^2_+}(\lambda),\xi_{\Eb^3_+}(\lambda),\xi_{\Eb^1_-}(\lambda),\xi_{\Eb^2_-}(\lambda),\xi_{\Eb^3_-}(\lambda),\xi_{\Omega}(\lambda),\xi_{\emptyset}(\lambda))
\end{align}
for each $\lambda$. These vectors are constrained as follows. From the properties of the trivial and null effect we have $\xi_{\Omega}(\lambda)=1$ and $\xi_{\emptyset}(\lambda)=0$ for all $\lambda$. We also have the positivity constraints $0\leq\xi_{\Eb^k_{\pm}}(\lambda)\leq 1$ and constraints stemming from the normalisation conditions $\xi_{\Eb^k_{+}}(\lambda)+\xi_{\Eb^k_{-}}(\lambda)=1$ for all $\lambda$. Finally, due to noncontextuality, the operational equivalence \eqref{opeqv_meas} implies the constraints 
\begin{align}\label{eq:xiequiv}
\xi_{\Eb^1_+}(\lambda)+\xi_{\Eb^2_+}(\lambda) +\xi_{\Eb^3_+}(\lambda) = \xi_{\Eb^1_-}(\lambda)+\xi_{\Eb^2_-}(\lambda) +\xi_{\Eb^3_-}(\lambda)\quad \forall \lambda. 
\end{align}
Combining these constraints we find that the vectors $\vec{\xi}(\lambda)$ can be characterised by the two-dimensional hexagon in Fig.~\ref{fig:mmtpolytope}. Each point in the polygon gives the values $\xi_{\Eb^1_+}(\lambda)$ and $\xi_{\Eb^2_+}(\lambda)$ with the other entries of $\vec{\xi}(\lambda)$ following from Eq.~\eqref{eq:xiequiv} and the normalisation conditions.
\begin{figure}
    \centering
    \includegraphics[scale=0.8]{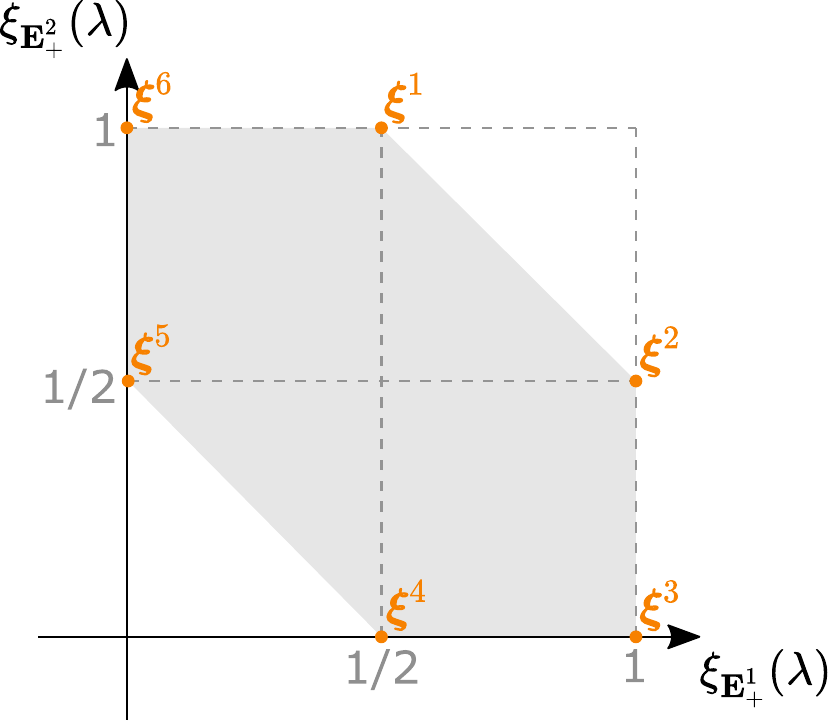}
    \caption{The six extremal vectors $\vec{\xi}_i$ can be characterised by the first two values, $\xi_{\Eb_+^1}(\lambda)$ and $\xi_{\Eb_+^2}(\lambda)$ in \eqref{vecxi}.}
    \label{fig:mmtpolytope}
\end{figure}
We therefore have 
\begin{align}
    \vec{\xi}(\lambda) =\sum_{i=1}^{6} p_\lambda(i)\vec{\xi}^i
\end{align}
 where $p_\lambda(i)$ is a probability distribution for each $\lambda$. Denoting by $\xi^i_{\Eb}$ the element of $\vec{\xi}^i$ corresponding to the effect $\Eb$, the operational statistics for preparation density $s$ and effect $\Eb$ predicted by the ontological model are
 \begin{align}
    P(\Eb \vert s) = \int_\Lambda \text{d}\lambda \mu_s(\lambda)\sum_i p_\lambda(i)\xi^i_{\Eb} &=  \sum_i\left(\int_\Lambda \text{d}\lambda \mu_s(\lambda) p_\lambda(i)\right) \xi^i_{\Eb} \nonumber\\
    &=: \sum_i\nu_{s}(i)\; \xi^i_{\Eb}, \label{linprobs}
 \end{align}
Where we have defined $\nu_s(i) := \int_\Lambda \text{d}\lambda \mu_s(\lambda) p_\lambda(i)$. Note that we have
\begin{align}\label{linprobscon2}
\sum_i\nu_s(i)=1, \quad \nu_s(i)\geq 0,
\end{align}
and by noncontextuality and the operational equivalence \eqref{opeqv_prep},
\begin{align}\label{linprobscon}
    \nu_{1}(i)+\nu_{2}(i)=\nu_{3}(i)+\nu_{4}(i)=\nu_{5}(i)+\nu_{6}(i),
\end{align}
where $\nu_j=\nu_{s_j}$.
A decomposition \eqref{linprobs} satisfying \eqref{linprobscon2} and \eqref{linprobscon} is necessary and sufficient for the existence of a noncontextual ontological model, since the decomposition \eqref{linprobs} can itself be seen as such a model by viewing the index $i$ as an ontic state. Together, \eqref{linprobs}, \eqref{linprobscon2} and \eqref{linprobscon} therefore give us a characterisation of the set of possible noncontextual statistics for the preparation densities $s_j$ $(j=1,\cdots,6)$ that consists of linear equalities and inequalities in the finite set of variables $\nu_{j}(i)$. From these, it follows that the inequality \eqref{ineq} is satisfied by the operational statistics of any multi-task learning model. This can be proven via linear quantifier elimination methods as in \cite{allnc}; the list of all such inequalities for this set of operational equivalences can be found in \cite{allnc,vickymate}. Here we give a pen-and-paper proof of the inequality, which may itself be of interest to the contextuality community. 

First, using the notation $\nu_{s_i}\equiv\nu_i$ and using \eqref{linprobs}, \eqref{linprobscon2}, we express the left hand side of the inequality \eqref{ineq} in terms of the variables $\nu_j$, which gives 
\begin{align}\label{nuineq}
    &\sum_{j=2,3}\nu_1(j)-\nu_5(j)+\sum_{k=1,6}\nu_3(k)-\nu_5(k) \nonumber\\ &\quad\quad\quad+\frac12\left[\nu_1(1)+\nu_1(4)+\nu_3(2)+\nu_3(5)-\sum_{l=1,2,4,5}\nu_5(l)\right] +\frac{3}{2}.
\end{align}
Using the normalisation and non-negativity constraints we can upper bound the following three quantities in the square brackets as follows:
\begin{align}
    &\nu_1(1)+\nu_1(4)\leq1-\nu_1(2)-\nu_1(3),\\
    &\nu_3(2)+\nu_3(5)\leq1-\nu_3(1)-\nu_3(6), \quad\text{ and}\\[4pt]
    -&\sum_{l=1,2,4,5}\nu_5(l)\leq-\sum_{l=4,5}\nu_5(l)=\sum_{l=1,2,3,6}\nu_5(l)-1.
\end{align} 
Rearranging, we obtain the following upper bound for the expression in \eqref{nuineq}:
\begin{align}\label{eq:newbound}
    \frac12+\frac12\left[\nu_3(1)+\nu_1(2)+\nu_1(3)+\nu_3(6)-\sum_{j=1,2,3,6}\nu_5(j)\right]+\frac{3}{2}.
\end{align}
From the operational equivalences \eqref{linprobscon} and non-negativity we have 
\begin{align}
\sum_{j=1,2,3,6}\nu_5(j)+\nu_6(j)=&\nu_3(1)+\nu_4(1)+\nu_1(2)+\nu_2(2)+\nu_1(3)+\nu_2(3)+\nu_3(6)+\nu_4(6)\\
    \geq&\nu_3(1)+\nu_1(2)+\nu_1(3)+\nu_3(6).
\end{align}
It follows that 
\begin{align}
    1\geq\sum_{j=1,2,3,6}\nu_6(j)
    \geq\nu_3(1)+\nu_1(2)+\nu_1(3)+\nu_3(6)-\sum_{j=1,2,3,6}\nu_5(j).
\end{align}
Using this in \eqref{eq:newbound} we obtain the bound $\frac{5}{2}$, which proves the inequality \eqref{ineq}. 
\end{proof}

We now consider a multi-task model described in Theorem \ref{thm:main}. From \eqref{stime} and the definition of $\vec{u}_i(\eta)$ one finds that the left hand side of the inequality \eqref{ineq} is
\begin{align}
    3\eta +\frac{3}{2}(1-\eta)= \frac{3}{2}(\eta+1).
\end{align}
This results in a violation of the inequality for $\eta>\frac{2}{3}$, implying contextuality of the learning model and therefore proving the claim of Theorem \ref{thm:main}.  We note that \eqref{ineq} more constraints on the model that are not covered by Theorem \ref{thm:main} or the inequality \eqref{ineq} can be found in \cite{allnc, vickymate}. The advantage of the phrasing of Theorem \ref{thm:main} however is that it can be understood easily via the geometric perspective of Figure \ref{fig:thmfig}. We also present a proof for the zero-noise case ($\eta=1$) in Appendix \ref{app:zeronoise}. Although less general, this proof may be useful to gain more intuition on the limitations of noncontextual models.

\section{Learning the rock, paper, scissors game}\label{sec:rps}
We now introduce a specific toy data set and learning problem---based on learning the payoff behaviour of a zero-sum game---to which Theorem \ref{thm:main} applies. Aside from bringing life to Theorem~\ref{thm:main}, this problem will also be the subject of the numerical investigation of section \ref{sec:surrogate}.

\subsection{The game}\label{rpsdata}
The data set is based on the well known rock, paper, scissors (RPS) game. Our game features three players that play a variant of the RPS game via a referee (Figure \ref{fig:scenario}A). In every round, each player plays either rock ($\R$), paper ($\P$) or scissors ($\S$), and we denote the choice of actions of the three players by the vector $\vec{a}\in\{\R,\P,\S\}^3$. Each player also has a special action whose role we will describe shortly. For Player 1, the special action is $\R$, for Player 2 it is $\P$ and for Player 3 it is $\S$. Following the actions, each player either receives receives \euro 1 or must pay \euro 1 to the referee, which we refer to as their pay-off, and denote the three payoffs as $\vec{y}=(\pm1,\pm1,\pm1)$. The payoffs are decided probabilistically based on the actions of the players. More precisely, the payoff of Player $k$ is such that
\begin{align}\label{exppayoff}
    \mathbb{E}[Y^{(k)}] =\frac{\ell_k(\vec{a})}{2},
\end{align}
where $\ell_k(\vec{a})\in\{-2,-1,0,1,2\}$ is equal to the number of players that Player $k$ beat in that round minus the number of players that beat Player $k$. This rule ensures that if a player beats (or loses to) all others, they are sure to win (or lose) the round.  

The rules for who beats whom are as follows. If two players choose different actions, the rule is the well-known rule:
\begin{align}
    \R_k>\S_l, \quad \P_k>\R_l, \quad \S_k>\P_l \quad \forall\;k,l,
\end{align}
where $\R_k>\S_l$ means ``Player $k$ plays rock beats Player $l$ plays scissors''. If two players play the same action, one player beats the other if the common action is their special action. For example, if players 1 and 2 both play $\R$, then Player 1 beats Player 2 since $\R$ is Player 1's special action. These rules can be summarised as
\begin{align}
    \R_1>\R_k,\;\;k=2,3 \quad\quad \P_2>\P_k, \;\; k=1,3 \quad\quad \S_3>\S_k, \;\; k=1,2.
\end{align}
In any other case it is a draw. Note that this game satisfies a zero-sum condition. Mathematically, since we must have $\sum_k\ell_k(\vec{a})=0$ it follows from \eqref{exppayoff} that
\begin{align}\label{zerosumcon}
    \mathbb{E}[Y^1\vert\vec{a}]+\mathbb{E}[Y^2\vert \vec{a}]+\mathbb{E}[Y^3\vert \vec{a}]=0.
\end{align}

\begin{figure}
    \centering
    \includegraphics[scale=1.0]{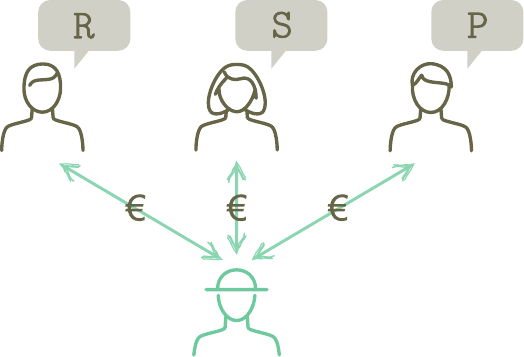}
    \caption{Three players play a variant of the rock, paper, scissors game with a referee. The referee determines (probabilistically) whether each player wins or loses the round based on their choice of action.}
    \label{fig:scenario}
\end{figure}

\subsection{The data set}
In the zero-sum game literature, it is common to introduce the concept of a \emph{strategy}. A strategy $\vec{x}(k)$ for Player $k$ is simply a list of probabilities that the player performs each possible action,
\begin{align}
    \vec{x}(k) = (P(A^{(k)}=\texttt{R}),P(A^{(k)}=\texttt{P}),P(A^{(k)}=\texttt{S})).
\end{align}
We denote by $\vec{x}\in D_{\vec{x}}$ the $3\times 3$ matrix whose $k^{\text{th}}$ row is $\vec{x}(k)$, where $D_{\vec{x}}$ is the space of valid strategy matrices, i.e.\ positive-valued matrices whose rows sum to $1$. Some of these strategies correspond to the six extremal behaviours in Figure \ref{fig:thmfig}. Consider for example the strategy
\begin{align}
\vec{x} = 
    \begin{pmatrix}
    1 & 0 & 0 \\
    0 & 0 & 1 \\
    0 & 0 & 1 
    \end{pmatrix},
\end{align}
i.e.\ a deterministic strategy where Player 1 plays $\texttt{R}$ and Player 2 and 3 play $\texttt{S}$. Following the rules of the game, the probabilities of a positive payoff for the three players are given by the vector $\vec{v}_1=(1,0,\frac{1}{2})$, which is one of the extremal behaviours in Figure \ref{fig:thmfig}. The other five extremal behaviours $\vec{v}_i$ are obtained by the deterministic strategies that correspond to the deterministic choice of actions $\texttt{SPP},\texttt{RPR},\texttt{SRS},\texttt{PPS}$ and $\texttt{RRP}$. 

The data set consists of pairs of strategies $\vec{x}$ and payoffs $\vec{y}$. We sample a number of strategies $\{\vec{x}_i\}$ following a distribution $
\mathcal{D}(\vec{x})$ that corresponds to generating a random positive matrix and normalising the rows (see Appendix \ref{app:numerics} for more details). The corresponding payoffs $\vec{y}_i$ are sampled according to the rules of the game given the strategy $\vec{x}_i$. That is, a choice of actions is sampled according to $\vec{x}_i$ and the corresponding pay-offs $\vec{y}_i$ is sampled via \eqref{exppayoff}. We therefore have the conditional distribution
\begin{align}\label{conprior}
    \mathcal{D}(\vec{y} \vert \vec{x}) = \sum_{\vec{a}\in\{\texttt{R},\texttt{P},\texttt{S}\}^3}P(\vec{a}\vert\vec{x})P(\vec{y}\vert\vec{a})
\end{align}
which defines our data distribution $\mathcal{D}(\vec{x},\vec{y})=\mathcal{D}(\vec{x})\mathcal{D}(\vec{y}\vert \vec{x})$. Note that the actions of the players are not explicit in the data, only the strategies $\vec{x}_i$. For a fixed $\vec{x}$, from \eqref{conprior} the conditional distribution $\mathcal{D}(\vec{y} \vert \vec{x})$ can be seen as a probabilistic mixture over the distributions $P(\vec{y}\vert \vec{a})$. It follows from \eqref{zerosumcon} that the data distribution $\mathcal{D}(\vec{x},\vec{y})$ satisfies the bias
\begin{align}\label{zerosumcon2}
    \mathbb{E}[Y^1\vert\vec{x}]+\mathbb{E}[Y^2\vert \vec{x}]+\mathbb{E}[Y^3\vert \vec{x}]=0,
\end{align}
for all $\vec{x}$, and is of the form described in Theorem \ref{thm:main}. Since this follows from the zero-sum property of the game, any data set constructed from a zero-sum game in the same fashion will satisfy \eqref{zerosumcon2}. 

\subsection{The learning problem}
We now describe a concrete problem that features the RPS data set. Consider a learning algorithm that has access to the RPS data $\{\vec{x}_i,\vec{y}_i\}$ and knows only that it was generated via a zero-sum game scenario, i.e.\ the bias \eqref{zerosumcon2} is assumed but the underlying game is unknown. The task we consider is to learn the pay-off behaviour of the game in the following sense. \vspace{5pt}
\begin{probo}[Learning the pay-off behaviour of a zero-sum game]{def:zerosumprob}
Consider a data set 
\begin{align}
\chi = \{\vec{x}_i,\vec{y}_i\}
\end{align}
of strategies $\vec{x}_i\in D_x$ and corresponding payoffs $\vec{y}_i\in\{\pm1\}^3$ sampled from a data distribution $\mathcal{D}(\vec{x},\vec{y})$ that respects the zero-sum bias \eqref{zerosumcon2}. Construct a multi-task model
\begin{align}
    \{\mathcal{P}_{\theta}^1(y\vert\vec{x}),\mathcal{P}_{\theta}^2(y\vert\vec{x}),\mathcal{P}_{\theta}^3(y\vert\vec{x})\}
\end{align}
that samples from the marginal payoff distributions $\mathcal{D}(y^{(k)}\vert \vec{x})$ for unseen strategies $\vec{x}$.
\end{probo}

We can now interpret Theorem \ref{thm:main} in the context of this learning problem. A multi-task model class that encodes the zero-sum bias contains models that are restricted to behaviours in the yellow hexagon of Figure \ref{fig:thmfig}B. Furthermore, if the model class is noncontextual, then the convex hull of these behaviours cannot contain the dashed hexagon. Thus, no noncontextual multi-task model that satisfies the zero-sum bias will be able to sample the correct payoff distributions for all possible strategies. In particular, if a test data set contains strategies close to the deterministic strategies corresponding to the choice of actions $\texttt{RSS}, \texttt{SPP},\texttt{RPR},\texttt{SRS},\texttt{PPS}$ and $\texttt{RRP}$, then, as described in Figure \ref{fig:thmfig}B, it will necessarily fail at sampling the correct payoff distribution for some strategies in the set, thus contributing to the expected generalisation error. We note that it is not too difficult to construct other data sets with the same structure by considering different scenarios in which a similar bias is present. In section \ref{sec:applications} we outline a number of these in more detail. 

\section{Contextuality of general learning models}\label{sec:relative}
Starting from the multi-task case, in this section we describe how one can arrive at a consistent definition of noncontextuality for single-task models (Definition \ref{def:ncbiasmodel}). We argue that in order to do this, it is more natural to think of contextuality relative to a particular bias of the model, rather than a property of the model as a whole. This is quite a subtle point that is important if trying to use our framework beyond this work. However, it is not necessary for understanding the rest of the paper, and some readers may wish to skip the section and return to it at a later time.

\subsection{Connecting multi-task and single-task models}
The definition of a noncontextual multi-task learning model (Definition \ref{def:opcontmodel}) considers all preparations and effects related to a user that interacts with the model at inference time. While this is a suitable perspective for multi-task problems such as \ref{def:zerosumprob}, if one wishes to define a consistent notion of noncontextuality that applies to more general, single-task models, considering the full set of preparations and effects can be overkill. To understand this issue, consider the modification to Problem \ref{def:genproblem} for $m=3$, in which the task is to learn the joint distribution of labels rather than the marginal distributions.  \vspace{5pt}
\begin{probo}[Learning a joint model with linear bias]{jointprob}
Consider a data set $\chi = \{\vec{x}_i,\vec{y}_i\}$ of input data $\vec{x}_i\in \mathbb{R}^d$ and label vectors $\vec{y}_i\in\{\pm1\}^3$ sampled from a data distribution $\mathcal{D}(\vec{x},\vec{y})$ satisfying 
\begin{align}\label{zerosumn}
    \mathbb{E}[Y^{(1)}\vert \vec{x}]+\mathbb{E}[Y^{(2)}\vert \vec{x}]+\mathbb{E}[Y^{(3)}\vert \vec{x}] = 0. 
\end{align}
Construct a model
    $\mathcal{P}_{\theta}(\vec{y}\vert\vec{x})$
that samples from the distribution $\mathcal{D}(\vec{y}\vert \vec{x})$ for unseen $\vec{x}$.
\end{probo}

We will call the model $\mathcal{P}_{\theta}(\vec{y}\vert\vec{x})$ a \emph{joint model} in order to contrast it to the multi-task model that is concerned only with the marginal statistics. A joint model can clearly be used to construct an associated multi-task model $
    \{\mathcal{P}_{\theta}^1(y\vert\vec{x}),\mathcal{P}_{\theta}^2(y\vert\vec{x}),\mathcal{P}_{\theta}^3(y\vert\vec{x})\}$
via three coarse-grainings that associate each element of the output label $\vec{y}$ to different tasks. In this way, a sufficiently expressive joint model can trivially simulate any multi-task model, and in practice one can use the joint model class to tackle Problem \ref{def:genproblem}. Suppose that a multi-task model constructed in this way satisfies the same operational equivalences \eqref{opeqv_prep} and \eqref{opeqv_meas} as before, and is contextual by virtue of violating the inequality \eqref{ineq}. Since this multi-task model is contextual and is obtained via a trivial post-processing of the joint model, it is natural to conclude that there is some form of contextuality present in the original joint model. However, if we blindly adopt Definition \ref{def:opcontmodel}, this is not necessarily the case. 

To see this, consider a user with access to the joint model $\mathcal{P}(\vec{y}\vert\vec{x})$. In the corresponding operational scenario the preparations are the same as before, however since we have a single task the effects are now
\begin{align}
    \{\textbf{E}_{+++},\Eb_{++-},\Eb_{+-+},\Eb_{+--},\Eb_{-++},\Eb_{-+-},\Eb_{--+},\Eb_{---}\}\cup \{\Omega,\emptyset\},
\end{align}
corresponding to the eight possible values of the label vector $\vec{y}$. Since the coarse-grained multi-task model satisfies the operational equivalence \eqref{opeqv_meas} there is an operational equivalence between the coarse-grained effects
\begin{align}
    \Eb_\pm^1 = \sum_{i,j\in\{+,-\}}\Eb_{\pm ij}, \quad \Eb_\pm^2 = \sum_{i,j\in\{+,-\}}\Eb_{i\pm j}, \quad \Eb_\pm^3 = \sum_{i,j\in\{+,-\}}\Eb_{ij\pm}
\end{align}
that correspond to observing a label $\pm1$ for each element of $\vec{y}$. We therefore have 
\begin{align}
    \Eb_+^1+\Eb_+^2+\Eb_+^3 \sim \Eb_-^1+\Eb_-^2+\Eb_-^3
\end{align}
as before. 

The issue arises when one considers the preparation equivalences $\eqref{opeqv_prep}$. 
Since the coarse-grained multi-task model satisfies $\eqref{opeqv_prep}$, it follows that 
\begin{align}\label{op1}
    P(\Eb_\pm^k\vert\half s_1+\half s_2)=    P(\Eb_\pm^k\vert \half s_3+\half s_4)
=    P(\Eb_\pm^k\vert \half s_5+\half s_6)
\end{align}
for $k=1,2,3$, i.e., the mixtures of preparations are indistinguishable if one observes marginal statistics only. However, in order to be operationally equivalent, the preparations must give identical statistics for all effects:
\begin{align}\label{op2}
    P(\Eb_{ijk}\vert \half s_1+\half s_2)=    P(\Eb_{ijk}\vert \half s_3+\half s_4)=    P(\Eb_{ijk}\vert \half s_5+\half s_6),\quad i,j,k\in\{\pm1\}.
\end{align}
Even if \eqref{op1} is satisfied, it may be the case that \eqref{op2} is not. This will be the case, for example, when the preparation densities $\frac{1}{2}s_1+\frac{1}{2}s_2$, $\frac{1}{2}s_3+\frac{1}{2}s_4$, and $\frac{1}{2}s_5+\frac{1}{2}s_6$ have identical marginal statistics, but differ in their joint statistics (see Appendix \ref{app:example} for an explicit example). In such cases we cannot conclude that
\begin{align}
   \frac{1}{2}s_1+\frac{1}{2}s_2\sim\frac{1}{2}s_3+\frac{1}{2}s_4\sim\frac{1}{2}s_5+\frac{1}{2}s_6.
\end{align}
As a result, the ontological model will lack the constraints implied by these equivalences, and we may be unable to prove contextuality\footnote{In fact, in quantum models, if one cannot identify any operational equivalences related to the preparations, there is always a noncontextual model \cite{gencon}.}. Considering the full set of preparations and effects therefore seems problematic in order to arrive at a consistent notion of contextuality across the single-task and multi-task settings. 

\subsection{Contextuality relative to a bias}
The moral of this story is that it may be more natural to think about contextuality relative to a particular bias of the model, rather than as a property of the model as a whole. Indeed, when studying contextuality it is always necessary to specify the procedures and observations that most accurately reflect the spirit of the nonclassicality one is wishing to probe, and different choices can lead to different conclusions \cite{mischa}. The solution to this conundrum therefore comes from considering contextuality with respect to those preparations and effects that play a role in a specific bias. In this way, one can identify a type of contextuality present in the operational statistics that is related to the non-existence of a noncontextual ontological model for these preparations and effects only. 

In particular, we may consider a general learning model $\mathcal{P}_
\theta(\vec{y}\vert\vec{x})$ with associated operational scenario $\{\texttt{Preps},\texttt{Effects},P(\Eb\vert\Sb)\}$ for which there exists an inductive bias that can be written as
\begin{align}\label{biasbias}
    \sum_{\Eb_i\in \mathfrak{E}_1} p_i P(\Eb_i\vert\Sb_{\vec{x}}) = \sum_{\Eb_j\in \mathfrak{E}_2} q_j P(\Eb_j\vert\Sb_{\vec{x}}) \quad \forall \Sb_{\vec{x}}
\end{align}
where $\mathfrak{E}_1, \mathfrak{E}_2 \subseteq \texttt{Effects}$ and $p_i, q_j$ are convex weights. This implies the operational equivalence 
\begin{align}
    \sum_{\Eb_i\in \mathfrak{E}_1}p_i\Eb_{i} \sim \sum_{\Eb_j\in \mathfrak{E}_1}q_j\Eb_j.
\end{align}
We can now define noncontextuality relative to the bias as follows. 
\begin{defo}[Noncontextual learning model relative to a bias]{def:ncbiasmodel}
Consider a learning model $\mathcal{P}(\vec{y}\vert\vec{x})$ with associated operational scenario 
\begin{align}
    \{\texttt{Preps},\texttt{Effects},P(\Eb \vert \Sb)\}
\end{align}
that satisfies an inductive bias of the form \eqref{biasbias}. The learning model is noncontextual relative to the bias \eqref{biasbias} if the operational statistics of the operational scenario
\begin{align}
\{\texttt{Preps},\mathfrak{E}_1\cup \mathfrak{E}_2,P(\Eb\vert \Sb)\}
\end{align}
admits a noncontextual ontological model.
\end{defo}
Using this definition, if the coarse-grained multi-task model is contextual in the sense of Definition \ref{def:opcontmodel}, the joint model is contextual relative to the bias \eqref{bias}, which unifies the multi-task and single-task approaches. We note that one could also consider biases that are satisfied by a subset of preparations, and consider a similar definition to the above involving these preparations only, but we do not consider this in this work.

\section{Encoding inductive bias into quantum learning models}\label{sec:encodingbias}
Until now, we have only been concerned with abstract representations of learning models in the form of preparations and effects of an operational scenario. In this section we present methods to construct  classes of multi-task quantum machine learning models $\{\{\mathcal{P}^k_\theta(y\vert\vec{x})\}_{k=1,\cdots,m}\}_\theta$ with $\vec{x}\in\mathbb{R}^{d}$ and $y\in\{-1,+1\}$ that encode the inductive bias
\begin{align}\label{linbias}
\mathbb{E}[Y^{(1)}\vert \vec{x}] + \mathbb{E}[Y^{(2)}\vert \vec{x}] + \cdots + \mathbb{E}[Y^{(m)}\vert \vec{x}] = 0 \quad \forall \vec{x}
\end{align}
for all parameters $\theta$ in the class. We consider a bias of the above form due to its simplicity and relevance to this work, but much of what we say can be straightforwardly extended to biases expressed as general linear combinations of the probabilities $\mathcal{P}_{\theta}^k(y\vert\vec{x})$ for arbitrary finite label dimension, and even to certain forms of nonlinear biases (as we explain in Section \ref{sec:nonlinear}). We remind the reader that since quantum theory is known to be contextual in the generalised sense \cite{gencon,vickymate}, the constraints on noncontextual models related to this bias do not apply to quantum models.

The type of quantum models we consider correspond to parameterised quantum circuits \cite{pqc1},
\begin{align}
    \mathcal{P}^k_{\vec{\theta}}(y\vert\vec{x}) = \bra{\psi_0}U^\dagger(\vec{x},\vec{\theta}) \mathcal{M}_{y\vert k} U(\vec{x},\vec{\theta})\ket{\psi_0},
\end{align}
where $\ket{\psi_0}$ is some $n$-qubit initial state, $U$ is a parameterised unitary that depends on the input data and some trainable parameters $\vec{\theta}$, and  $\{\{\mathcal{M}_{+1\vert k},\mathcal{M}_{-1\vert k}\}\}_k$ is a set of two-outcome measurements that sample the label for task $k$. Denoting by $\mathcal{O}_k = \mathcal{M}_{+1\vert k}-\mathcal{M}_{-1\vert k}$ the corresponding observable for task $k$, from \eqref{linbias} our aim is to find choices of $\mathcal{O}_k$, $\ket{\psi_0}$ and $U$ such that 
\begin{align}\label{biasq}
\sum_k \bra{\psi_0}U^\dagger(\vec{x},\vec{\theta}) \mathcal{O}_{k} U(\vec{x},\vec{\theta})\ket{\psi_0} = \bra{\psi_0}U^\dagger(\vec{x},\vec{\theta})\mathcal{H} U(\vec{x},\vec{\theta})\ket{\psi_0} =  0 \quad \forall \vec{x},\vec{\theta},
\end{align}
where we call $\mathcal{H}=\sum_k \mathcal{O}_{k}$ the \emph{bias operator}. Note that from the above, $\mathcal{H}$ has to have at least one non-positive eigenvalue.  We explore two methods that we call the state-based and measurement-based approaches. As we will see, the linearity of quantum theory plays an important role, and provides a rather natural way of encoding such biases. 

\begin{figure}
    \centering
    \includegraphics{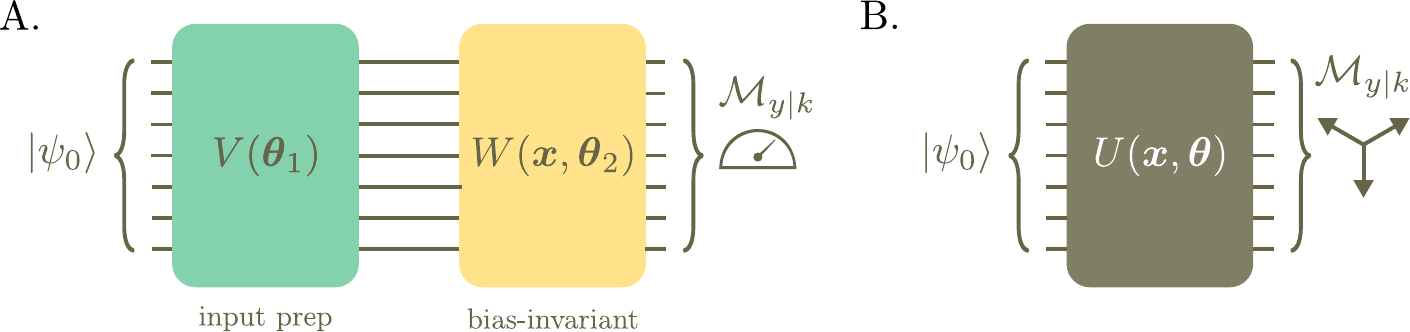}
    \caption{Two methods for encoding the inductive bias \eqref{linbias} into a quantum machine learning model. A.\ The state-based approach. The first unitary creates states that satisfy the bias for the chosen measurement operators $\mathcal{M}_{y\vert k}$. The second unitary is a bias-invariant layer that performs rotations that preserve the bias for all input states. B.\ The measurement-based approach. Here, an arbitrary input state and unitary can be chosen, but the measurement observables are chosen to have a specific structure that enforces the bias.}
    \label{fig:ansatze}
\end{figure}

\subsection{State-based approach}\label{state}
The first step of this approach is to fix a choice of observables $\{\mathcal{O}_k\}$. A simple choice we will use in the next section is to take $\mathcal{O}_k = \sigma_z^{(k)}$, so that a computational basis measurement on the $k^{\text{th}}$ qubit samples the label for task $k$.  Given such a choice,  the question then is how to construct $U$ and $\ket{\psi_0}$ such that \eqref{biasq} holds. The method we propose is to build the quantum model from a sequence of two subcircuits $V(\vec{\theta})$ and $W(\vec{\vec{x},\theta})$ (see Figure \ref{fig:ansatze}A). The first circuit is a parameterised unitary that we call the \emph{input preparation unitary}, whose role is to prepare a state for the next circuit that satisfies the bias, i.e.\ such that $\langle\mathcal{H}\rangle=0$. The second circuit is a \emph{bias-invariant layer}, meaning that it preserves the quantity $\langle\mathcal{H}\rangle$ for arbitrary input states from the first layer. 

We first focus on the second circuit $W(\vec{\vec{x},\theta})$. In order to be bias-invariant, unitaries in this circuit must belong to the set
\begin{align}\label{G}
    G = \{ U \;\vert\;  \bra{\psi}U^\dagger\mathcal{H}U\ket{\psi}=\bra{\psi}\mathcal{H}\ket{\psi} \; \forall \ket{\psi} \},
\end{align}
which forms a matrix Lie group under matrix multiplication. The associated matrix Lie algebra\footnote{Here we use the `physicist' convention where we include a factor of $i$ so that the matrices $X$ are Hermitian; in the mathematics literature the $i$ is usually absorbed into the matrix $X$.} is equivalent to the vector space of matrices that exponentiate to elements of $G$:
\begin{align}
    \mathfrak{g} = \{ X \;\vert\; \exp{( itX)}\in G \; \forall t\in\mathbb{R}\}.
\end{align}
For the group $G$, this is just equal to those Hermitian operators that commute with $\mathcal{H}$,
\begin{align}\label{liecommute}
    \mathfrak{g} = \{ X \;\vert\; X^\dagger = X, [X,\mathcal{H}]=0\}.
\end{align}
It follows that the bias $\langle \mathcal{H}\rangle$ is conserved by any data-encoding unitary of the form $W(\vec{x}) = \exp{(i X f(\vec{x}}))$ and any trainable unitary of the form $ W(\vec{\theta}) =  \exp{(i X g(\vec{\theta}))}$, where $X\in\mathfrak{g}$ and $f$, $g$ are arbitrary real functions. Any number of unitaries of this form can therefore be combined in sequence to construct a bias-invariant circuit $W(\vec{\vec{x},\theta})$. The main challenge here is determining the set \eqref{liecommute} and judiciously choosing the generators $X$ from which $W$ is defined. For choices of $\mathcal{O}_k$ with a lot of structure this may be relatively easy however. For example, for the choice $\mathcal{O}_k = \sigma_z^{(k)}$, one finds that any generator $X$ that is diagonal in the computational basis, as well as those of the form $X=\sigma_x^{(i)}\sigma_x^{(j)}+\sigma_y^{(i)}\sigma_y^{(j)}$, $i,j\in\{1,\cdots,m\}, i\neq j$ commute with $\mathcal{H}$. These may be combined linearly to create new generators that also commute with $\mathcal{H}$.

We now turn to the input preparation circuit $V(\vec{\theta})$. This should prepare a state that satisfies the bias, i.e.\ it consists of unitaries in the set 
\begin{align}
    G' = \{ U \;\vert\;  \bra{\psi_0}U^\dagger\mathcal{H}U\ket{\psi_0}=0 \}.
\end{align}
In general this does not form a group (since it is not closed) and therefore may be more difficult to characterise than $G$. Note that even if $\ket{\psi_0}$ satisfies the bias, unitaries in $G'$ do not necessarily commute with $\mathcal{H}$. An example of such a unitary for the case $\mathcal{O}_k = \sigma_z^{(k)}$ and $\ket{\psi_0}=\ket{+}^{\otimes n}$ is
\begin{align}
    V(\theta)=\exp{(i\sigma_x^{(1)} \theta)}\exp{(-i\sigma_x^{(2)} \theta)}
\end{align}
for any $\theta\in(0,\pi]$. We are unaware of general methods to construct parameterised unitaries in $G'$, so at the moment we rely on ad-hoc approaches like the above. 
%

\subsection{Measurement-based approach}
The second approach (Figure \ref{fig:ansatze}B) is based on encoding the bias into the structure of the measurement operators rather than the parameterised unitary. In this case, one can consider an arbitrary input state $\ket{\psi_0}$ and arbitrary parameterised unitary $U(\vec{x},\vec{\theta})$. The observables however must be chosen so that 
\begin{align}\label{ops}
    \mathcal{O}_1 + \mathcal{O}_2 + \cdots + \mathcal{O}_m = \mathcal{H} = 0. 
\end{align}
From the linearity of the trace it follows that this choice ensures \eqref{linbias} is satisfied. An example of observables that satisfy \eqref{ops} for a qubit and $m=3$ are any three equally spaced `trine' observables in a plane of the Bloch sphere, e.g.,  $\mathcal{O}_k = \vec{v}_k\cdot \vec{\sigma}$ with $\vec{v}_k=(\cos \theta_k,\sin \theta_k,0)$ and $\theta_k=2k\pi/3$. However, to the best of our knowledge, observables of this form for larger dimension are not known and further research is required in this direction. Here, one may expect that mathematical tools from other constructions of symmetric measurements such as mutually unbiased bases and SIC-POVMs to be useful \cite{bengtsson2017geometry}. In Appendix \ref{matemagic} we give a previously unknown construction for the case $m=3$ for all even dimensions that satisfies \eqref{ops}. 

Finally we comment that for both the state and measurement based approaches, it may be beneficial to construct models that sample from classical probabilistic mixtures of such circuits by combining them with stochastic neural network models (see e.g.\ \cite{verdon2019quantum}), but we do not consider this here. 

\subsection{Nonlinear biases}\label{sec:nonlinear}
The same techniques can also be used to encode nonlinear biases if they can be expressed through a bias operator as $\langle \mathcal{H}\rangle=0$ in the quantum model. Consider for example a bias of the form
\begin{align}\label{nonlinbias}
\mathbb{E}[f_1(Y^{(1)})\vert \vec{x}] + \mathbb{E}[f_2(Y^{(2)})\vert \vec{x}] + \cdots + \mathbb{E}[f_m(Y^{(m)})\vert \vec{x}] = 0 \quad \forall \vec{x}
\end{align}
where the $f_k$ are potentially nonlinear functions of the labels. If the observables $\mathcal{O}_k$ are of the form $\mathcal{O}_k=\sum_{y} y\mathcal{M}_{y\vert k}$ with $\mathcal{M}_{y\vert k}$ the projector corresponding to label $y$ of task $k$, then the left hand side of the above is given by \eqref{biasq} where 
\begin{align}
\mathcal{H} = f_1(\mathcal{O}_1)+f_2(\mathcal{O}_2)+\cdots +f_m(\mathcal{O}_m),
\end{align}
and $f_k(\mathcal{O}_k)=\sum_{y} f_k(y)\mathcal{M}_{y\vert k}$. Thus, one can use the state-based or measurement-based approach described with this bias operator in order to construct model classes that encode \eqref{nonlinbias}. Of course, the difficult task to find the corresponding generators $X$ in \eqref{liecommute} or a collection of measurements satisfying \eqref{ops} for a given nonlinear bias. It would be interesting to understand which biases admit efficient model constructions in this way. 

\section{Outperforming classical surrogates}\label{sec:surrogate}

In this section we construct a class of quantum models that encode the bias \eqref{linbias} for $m=3$, using the state-based approach of Section \ref{state}. We then train the model class using the RPS data set described in Section \ref{sec:rps}. In order to investigate the effect of the encoded bias, we compare the quantum models to a class of classical surrogate models recently introduced in \cite{surrogates}. As we will see, the quantum model class achieves a lower generalisation error than its associated surrogate model class, which we claim to be due to the ability of the quantum class to encode the relevant inductive bias. We stress that we are not claiming that this is evidence of `quantum advantage' since the quantum models we consider are small enough to be easily simulated on a laptop computer. Code used to generate the plots of this section can be found at \url{https://github.com/XanaduAI/Contextual-ML}.

\subsection{The learning problem}
The learning problem we consider is that of learning the payoff behaviour of an unknown zero-sum game, as described in Problem \ref{def:zerosumprob}. We use the RPS data set
\begin{align}
    \chi = \{\vec{x}_i,\vec{y}_i\},
\end{align}
containing $1500$ strategies $\vec{x}_i\in D_x$ and payoffs $\vec{y}_i$. Since from normalisation the last entry of each row of $\vec{x}$ is linearly dependent on the first two, we compress each $\vec{x}_i$ into a matrix of size $3\times 2$, and scale the data by $\frac{\pi}{2}$ so that the input data points correspond to the matrix
\begin{align}\label{newdata}
    \vec{x} = \frac{\pi}{2}\begin{pmatrix}
    P(A_1=\texttt{R}) & P(A_1=\texttt{P})-P(A_1=\texttt{S}) \\
    P(A_2=\texttt{P}) & P(A_2=\texttt{S})-P(A_2=\texttt{R}) \\
    P(A_3=\texttt{S}) & P(A_3=\texttt{R})-P(A_3=\texttt{P})
    \end{pmatrix} = \begin{pmatrix}
    x_1 & x_4 \\
   x_2 & x_5 \\
   x_3 & x_6
    \end{pmatrix}. 
\end{align}

\begin{figure}
    \centering
    \includegraphics[width=\textwidth]{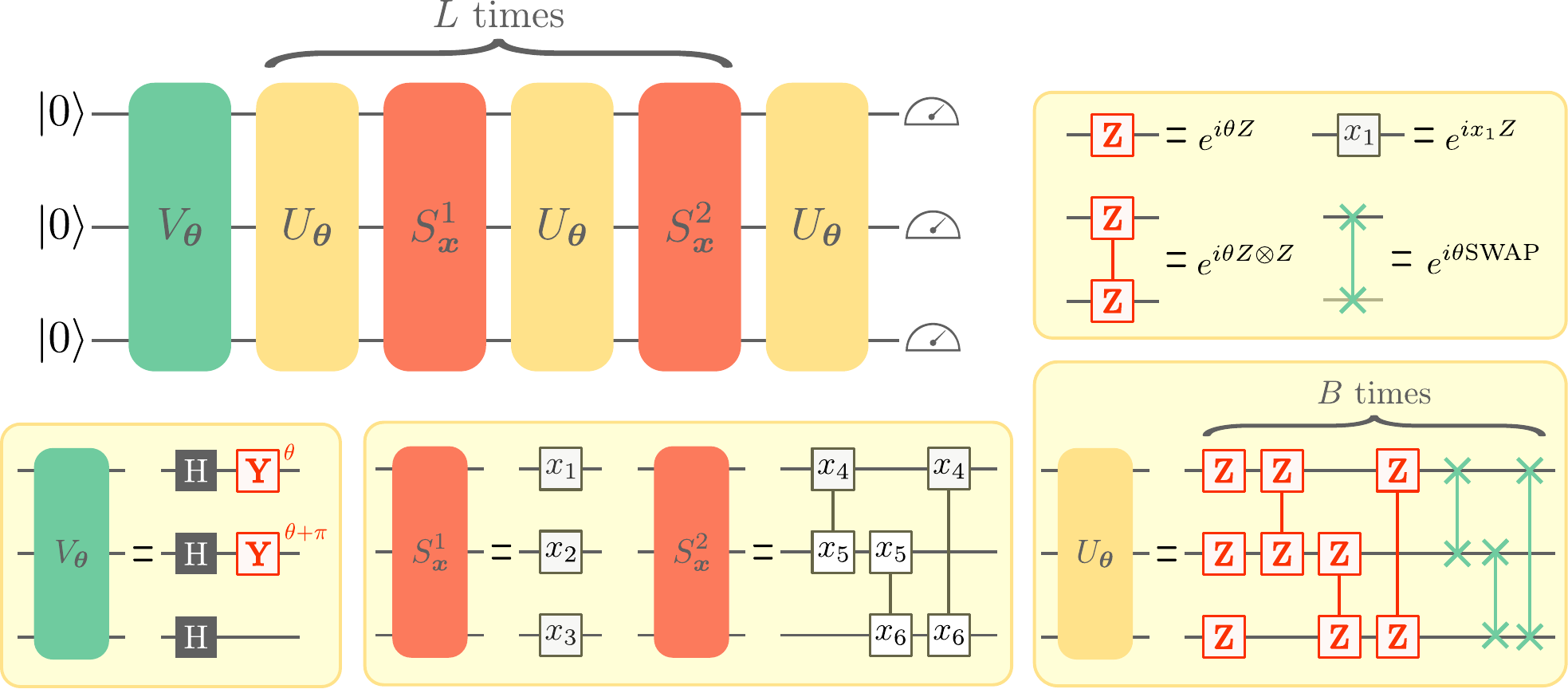}
    \caption{The three-qubit quantum model class we use in our numerical investigation that encodes the bias \eqref{linbias} using the state-based approach of section \ref{sec:encodingbias} (here we use the notation $Z\equiv\sigma_z$). The bias operator for this ansatz is $\mathcal{H}=\sigma_z^1+\sigma_z^2+\sigma_z^3$.  The input preparation unitary $V_{\theta}$ prepares a state such that $\langle \mathcal{H}\rangle=0$, where single-qubit Y rotations on the first and second qubit are correlated to have a difference of $\pi$. The following layers are bias-invariant layers. Each parameterised gate appearing in a unitary $U_{\vec{\theta}}$ has its own parameter (i.e.\ the unitaries $U_{\vec{\theta}}$ can have different parameter values).  Boxes connected by a line in $S^2_{\vec{x}}$ and $U_{\vec{\theta}}$ correspond to parameterised rotations $\exp(i x_k x_l \sigma_z\otimes \sigma_z)$ and $\exp (i \theta \sigma_z\otimes \sigma_z)$ respectively, and SWAP is the usual two qubit swap operator.}
    \label{fig:numerics}
\end{figure}

\subsection{The quantum model class}
We consider a simple class of three-qubit quantum models, described in Figure \ref{fig:numerics}, which encode the bias using the state-based method of section \ref{sec:encodingbias}. This class contains models of the form $\mathcal{P}_{\theta}(\vec{y}\vert\vec{x})$ which we view as a multi-task model $\{\mathcal{P}^1_{{\theta}}(y\vert\vec{x}),\mathcal{P}^2_{{\theta}}(y\vert\vec{x}),\mathcal{P}^3_{{\theta}}(y\vert\vec{x})\}$ by associating each element of $\vec{y}$ with a separate task. Following \cite{schuld2021effect}, the expectation value of the final $\sigma_z$ measurement on any of the qubits is equivalent to a truncated Fourier series of the form
\begin{align}\label{fourier}
    f_{\theta}(\vec{x})=\sum_{\vec{\omega}\in\Omega} c_{\vec{\omega}}(\theta)e^{i\vec{\omega}\cdot\vec{z}(\vec{x})},
\end{align}
where $\vec{z}(\vec{x})=(x_1,x_2,x_3,x_4x_5,x_4x_6,x_5x_6).$ The frequency spectrum $\Omega$ is determined by the eigenspectra of the data-encoding gates. In our case one finds
\begin{align}\label{omegas}
    \Omega_L = \{(\omega_1,\omega_2,\cdots,\omega_6)\;\vert\; \omega_i\in\{-L,-L+1,\cdots,L-1,L\}\;\forall i\},
\end{align}
giving a total of $(2L+1)^6$ vectors $\omega$. The coefficients $c_{\vec{\omega}}(\theta)$ are determined by the structure of the rest of the circuit; note that we have $c_{\vec{-\omega}}=c^*_{\vec{\omega}}$ so that $f$ is real-valued. The model therefore samples labels $\vec{y}$ such that 
\begin{align}\label{qfourier}
    (\mathbb{E}[Y^{(1)}\vert \vec{x}], \mathbb{E}[Y^{(2)}\vert \vec{x}], \mathbb{E}[Y^{(3)}\vert \vec{x}]) = (f^1_{\theta}(\vec{x}),f^2_{\theta}(\vec{x}),f^3_{\theta}(\vec{x})),
\end{align}
where each $f^k_{\theta}$ is of the form \eqref{fourier}. Note that each of the generators of the gates in the circuit commute with the bias operator $\mathcal{H}=\sigma_{z}^{(1)}+\sigma_{z}^{(2)}+\sigma_{z}^{(3)}$ and that the input preparation unitary $V_{\theta}$ always prepares a state such that $\langle \mathcal{H} \rangle=0$. The values of the coefficients of the three Fourier series $f^k_{\theta}(\vec{x})$ are therefore correlated by the circuit structure in such a way that the zero-sum bias is always satisfied.

\begin{figure}
    \centering
    \includegraphics[width=\textwidth]{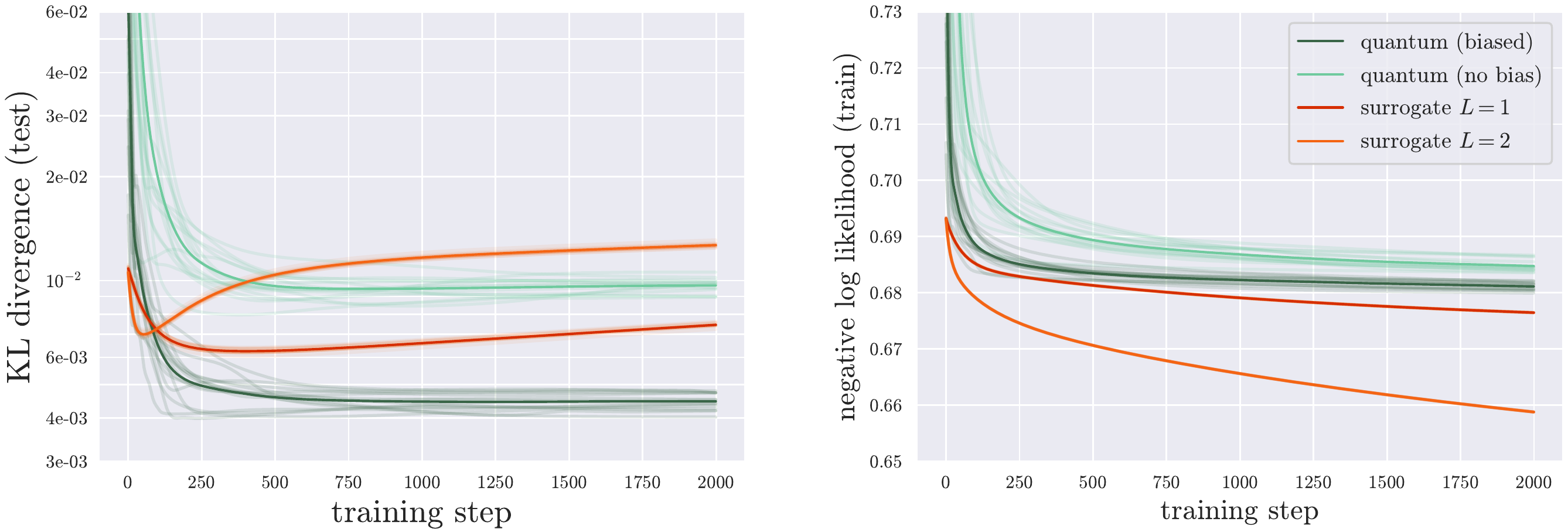}
    \caption{left: The average KL divergence estimated via \eqref{avkl} for the four models we consider. The solid lines denote the average over $20$ random initialisations of the weights (light curves). The best performance is achieved by the biased quantum model. Both surrogate models suffer from overfitting, as can be seen from the increase in KL divergence, and non-convergence of the negative log likelihood \eqref{likelihood} (right).}
    \label{fig:plots}
\end{figure}

\subsection{Classical surrogates of quantum models}
The idea of constructing a classical surrogate model of a quantum model was recently introduced in \cite{surrogates} to study the inductive bias of quantum models for regression tasks, i.e.~models of the form $f_{\theta}(\vec{x})\in\mathbb{R}$. A surrogate model class is simply a linear model in the same Fourier features of the quantum model. For example, consider a quantum model class of the form \eqref{fourier}. The surrogate class is the parameterised truncated Fourier series
\begin{align}\label{linmodel}
    g_{\vec{a},\vec{\beta}}(\vec{x}) = \sum_{\vec{\omega}\in\Omega}\alpha_{\vec{\omega}}\cos(\vec{\omega}\cdot\vec{z(\vec{x}}))+\beta_{\vec{\omega}}\sin(\vec{\omega}\cdot\vec{z}(\vec{x})),
\end{align}
where the $\alpha_{\vec{\omega}}$ and $\beta_{\vec{\omega}}$ are real-valued trainable parameters and $\Omega$ is the same as in \eqref{omegas}. We have chosen this parameterisation rather than the exponential form since we know the model must produce a real number and can therefore avoid using complex numbers. The surrogate model class is larger than the quantum model class since it contains all models in \eqref{fourier}, however it lacks the specific bias of the quantum class that is encoded in the structure of the Fourier coefficients. Interestingly, surrogate model classes often perform better than generic quantum models for regression on a range of data sets \cite{surrogates}.

Since our problem is not a regression problem but involves sampling the payoff labels, we need to define a natural notion of a classical surrogate to compare the quantum model \eqref{qfourier} against. The logic we will follow is as above: view each $f^k_{\theta}(\vec{x})$ of the quantum model as a truncated Fourier series and construct a surrogate model by ignoring the structure in the coefficients. To this end, we consider a collection of three linear models
\begin{align}
 (g^1_{\vec{\alpha}_1,\vec{\beta}_1}(\vec{x}),g^2_{\vec{\alpha}_2,\vec{\beta}_2}(\vec{x}),g^3_{\vec{\alpha}_3,\vec{\beta}_3}(\vec{x})),
\end{align}
where each $g^k_{\vec{\alpha}_k,\vec{\beta}_k}(\vec{x})$ is of the form \eqref{linmodel} and has the spectrum $\Omega_L$. In analogy to \eqref{qfourier} we would like to interpret these values as expectation values $\mathbb{E}[Y^{(k)}\vert \vec{x}]$, however this is not immediately possible since $g^k_{\vec{\alpha}_k,\vec{\beta}_k}$ is not bounded in $[-1,1]$. To remedy this, we use a HardTanh function 
\begin{align}
    h( x ) = 
    \begin{cases}
    -1 & \text{if } x< -1\\
    x  & \text{if } -1 \leq x \leq 1 \\
    1 & \text{if } x>1
\end{cases}
\end{align}
so that payoffs are sampled from the model as 
\begin{align}
(\mathbb{E}[Y^{(1)}\vert \vec{x}], \mathbb{E}[Y^{(2)}\vert \vec{x}], \mathbb{E}[Y^{(3)}\vert \vec{x}]) =  (h(g^1_{\vec{\alpha}_1,\vec{\beta}_1}(\vec{x})),h(g^2_{\vec{\alpha}_2,\vec{\beta}_2}(\vec{x})),h(g^3_{\vec{\alpha}_3,\vec{\beta}_3}(\vec{x}))). 
\end{align}
This surrogate model class has access to the same frequency spectra as the quantum model class and by virtue of being less constrained, contains the quantum class as a specific subset of its parameters. The surrogate model does not however encode the zero-sum bias \eqref{linbias}. This is  because the constraint is related in a complex manner to the weights $\vec{\alpha}_k,\vec{\beta}_k$ and it is therefore not clear how to enforce it. Moreover, attempting to enforce the constraint via a postprocessing of the $g^k$'s seems cumbersome, and we were unable to identify a general method that would scale to larger label dimensions, unlike in the quantum case. We therefore expect that it is infeasible to encode this kind of bias into a linear Fourier features model, and indeed it is this difference that we expect to lead to a lower generalisation error in the quantum model.

\begin{figure}
    \centering
\includegraphics[width=\textwidth]{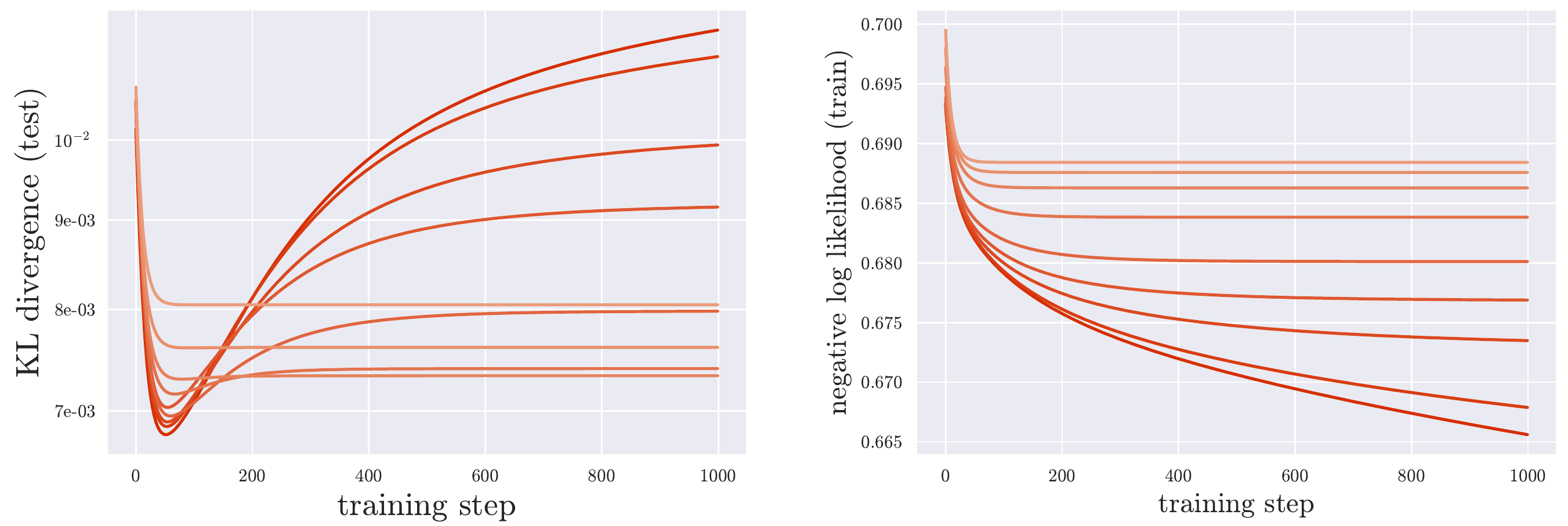}
    \caption{The effect of adding L2 regularisation to the cost function for the surrogate model for $L=2$ for different regularisation strengths (lighter curves imply stronger regularisation). This controls overfitting but does not result in better generalisation.}
    \label{fig:reg}
\end{figure}

\subsection{Training and evaluation}
To train the models we minimise the negative log likelihood of the labels given the strategies
\begin{align}\label{likelihood}
    \mathcal{L} = -\sum_{(\vec{x}_i,\vec{y}_i) \in \chi} \left(\log(\mathcal{P}^1_{\theta}(y_i^{(1)}\vert \vec{x}_i))+\log(\mathcal{P}^2_{\theta}(y_i^{(2)}\vert \vec{x}_i))+\log(\mathcal{P}^3_{\theta}(y_i^{(3)}\vert \vec{x}_i))\right)
\end{align}
using the Adam gradient-descent optimiser. Full details on the data sampling and training can be found in Appendix \ref{app:numerics}. We acknowledge that this training method does not scale well for quantum models since it requires the ability to efficiently estimate log probabilities, which becomes infeasible for large circuits. Indeed, a crucial challenge of quantum machine learning is to develop scalable methods that perform well on generative tasks, but we do not focus on this problem in this work. 

To evaluate the trained models we consider the sum of the Kullback–Leibler divergences between the trained models $\mathcal{P}^k_{\theta}(y\vert \vec{x})$ and the true data distributions $\mathcal{D}_k(y\vert \vec{x})\equiv \mathcal{D}(y^{(k)}\vert \vec{x})$, averaged over $\mathcal{D}(\vec{x})$:
\small
\begin{align}
    \mathbb{E}_{\mathcal{D}(\vec{x})} \left[\frac{1}{3}\sum_{k=1}^{3} D_{\text{KL}}(\mathcal{D}_k(Y\vert\vec{x})\vert\vert \mathcal{P}^k_{\theta}(Y\vert\vec{x})) \right] \label{kltest}
    = \mathbb{E}_{\mathcal{D}(\vec{x})} \left[\frac{1}{3}\sum_{k=1}^{3} \left(\sum_{y=\pm1}\mathcal{D}_k(y\vert\vec{x}) \log \left(\frac{\mathcal{D}_k(y\vert\vec{x})}{\mathcal{P}^k_{\theta}(y\vert\vec{x})}\right)  \right) \right].
\end{align}
\normalsize
%
We sample $N=10000$ strategies $\vec{x}_i^{\text{test}}$ via $\mathcal{D}(\vec{x})$ and estimate the average \eqref{kltest} with the sum 
\begin{align}\label{avkl}
    \frac{1}{3N}\sum_{\vec{x}_i^{\text{test}}} \sum_{k=1}^3 D_{\text{KL}}(\mathcal{D}_k(Y\vert\vec{x}_i^{\text{test}})\vert\vert \mathcal{P}_k(Y\vert\vec{x}_i^{\text{test}})),
\end{align}
which we calculate to machine precision using symbolic expressions of the model probabilities.

\subsection{Results}
In our study we consider four model classes: (i) The biased quantum model given in Figure \ref{fig:numerics} with $L=2$ and $B=2$, (ii) a generic quantum model with approximately the same number of parameters, the details of which we give in appendix \ref{app:ansatz}, (iii) a surrogate model corresponding to a $L=1$ quantum model, (iv) a surrogate corresponding to a $L=2$ quantum model. The results are given in Figure \ref{fig:plots}. As has been observed before \cite{surrogates}, the surrogate models are able to outperform a generic quantum model class. The best performance is achieved by the biased quantum model class however, which we conclude is due to the ability of the quantum model to naturally encode the zero-sum bias. Note that the surrogate models suffer severely from over-fitting, as can be seen from the decrease in the negative log likelihood and increase of generalisation error. This is due to the large number of parameters the surrogate models have relative to the quantum models. To try and counter this, we also optimised the surrogate models using a $L2$ weight penalty term added to the cost \eqref{likelihood}. Although this alleviates the problem of overfitting, we did not find that it improved generalisation performance (see Figure \ref{fig:reg} for the case $L=2$ and Figure \ref{fig:reg1} of the appendix for $L=1$). This suggests that the lower generalisation error of the quantum model over the surrogate originates from the inductive bias rather than overfitting in the surrogate model.

\section{Avenues for contextuality-inspired inductive bias}\label{sec:applications}
In this work we have focused on a single,  simple inductive bias and shown that contextuality plays a role in learning. Our ultimate goal however is to identify complex problems of practical interest for which contextuality is crucial, and for which the structure of quantum models makes them the model of choice. In this section, we suggest a number of scenarios with this goal in mind, and finish by discussing possible extensions to our framework. 

\subsection{Zero-sum games with larger payoff dimension}
Most real-life zero-sum game scenarios involve payoffs that take more values than only $\pm1$ as we considered here. For example, one could consider a zero-sum game scenario with $N$ players and conditional payoff distributions $\mathcal{P}(\vec{y}\vert\vec{x})$ with each element of $\vec{y}$ in $[-C,\cdots,C]$ for some integer $C$. From the zero-sum condition it follows that 
    $\sum_{k=1}^N \mathbb{E}(Y^{(k)}\vert\vec{x}) = 0
    $ $\forall\;\vec{x}$
as before. To see how this leads to a constraint on noncontextual models, we first shift the values of the payoffs by $C$ so that they are in the range $[0,2C]$. This does not change the nature of the problem since it is just a shift of the label data. The zero-sum condition then becomes 
\begin{align}\label{genzerosum}
    \sum_{k=1}^{N} \sum_{y=0}^{2C} \mathcal{P}^k_\theta(y\vert\vec{x})y = C \quad\forall \vec{x} \implies \frac{1}{N}\sum_{k=1}^{N} \sum_{y=0}^{2C} p_y\mathcal{P}^k_\theta(y\vert\vec{x}) = \frac{1}{N(2C+1)} \quad\forall \vec{x},
\end{align}
where $p_y= y/\sum_y y = \frac{y}{C(2C+1)}$ is a probability distribution. 

The right-hand side of the above allows us to identify an operational equivalence in the corresponding operational scenario of a multi-task model. Consider the following two effect densities: ($e_1$) A player is chosen at random. A $2C$-sided die with probabilities $p_i$ is rolled with outcome $y$ and the payoff of Player $k$ is observed, and success corresponds to the model returning the payoff value $y$; ($e_2$) The model is ignored and success is declared with probability $q=1/(N(2C+1))$. These two procedures correspond to the left and right sides of the equality in the right hand side of \eqref{genzerosum}. It follows that the two measurement procedures give identical statistics for all strategies. Since $e_2 = q\Omega +(1-q)\emptyset$ this implies the operational equivalence
\begin{align}
    \frac{1}{N}\sum_{k=1}^N \sum_{y=1}^{2C}p_y\Eb_{y\vert k} \sim q\Omega +(1-q)\emptyset,
\end{align}
where $\Eb_{y\vert k}$ is the effect corresponding to observing payoff $y$ of Player $k$. This leads to the constraint
\begin{align}\label{zerosumgenconstraint}
\sum_{k=1}^N \sum_{y=1}^{2C}\xi_{\Eb_{y\vert k}}( \lambda)p_y = \frac{1}{2C+1} \quad\forall \;\lambda
\end{align}
on the noncontextual ontological model. This will limit the expressivity of the model as before, however a clear picture of the structure of these constraints and how they affect learning is needed. We note that one need not necessarily consider only zero-sum games; similar constraints can be found in the same manner in any constant-sum game as well. 

\subsection{Conserved substances in networks}
Biases of the form \eqref{linbias} become relevant when modelling the behaviour of conserved substances in networks, which may have relevance to problems in biology (e.g.\ a network of cells sharing a conserved substance), population modelling (a network of locations sharing a conserved number of people/animals), manufacturing logistics (a network of factories sharing a conserved number of resources). To formalise these settings, consider a network consisting of $N$ cells that share some substance. We denote the amount of the substance in cell $j$ by $y_j$, which we assume to take discrete values. We also consider some parameters $\vec{x}$ (e.g.\ time, network connectivity) that affect the exchange of the substance throughout the network, and assume that the total amount of substance is conserved on average:
\begin{align}\label{cells}
    \mathbb{E}[Y_1\vert \vec{x}]+ \mathbb{E}[Y_2\vert \vec{x}]+\cdots +  \mathbb{E}[Y_N\vert \vec{x}] = C.
\end{align}
Consider a multi-task model whose aim is to model the amount of substance in a given cell given parameters $\vec{x}$. That is, it is a probabilistic model $\mathcal{P}^k_\theta(y\vert\vec{x})$ that from \eqref{cells} satisfies 
\begin{align}
    \sum_k \sum_y \mathcal{P}^k_\theta(y\vert\vec{x}) y = C.
\end{align}
This is of the same from as the left hand side of \eqref{genzerosum}, and so an analogous constraint to \eqref{zerosumgenconstraint} can be derived on the corresponding noncontextual ontological model. 

\subsection{Quantum control under a conserved quantity}
Consider a quantum control scenario where a multipartite intial state $\rho_0$ is evolved to a final state $\rho(\vec{x})$ and one of a number of measurements $\{\Eb_{y\vert k}\}$ is made with outcome $y=0,1,\cdots,M$. Here, $\vec{x}$ are some parameters of the control scenario, for example strengths of various interactions and interaction time.  Suppose there exists an operator of the form 
\begin{align}
    \mathcal{O} = \sum_{yk}\Eb_{y\vert k} g_y
\end{align}
with $g_y\geq 0$ that is conserved during the control process, i.e. $\langle \mathcal{O}\rangle =C \;\forall \vec{x}$. For example, $ \mathcal{O}$ could be the energy of the system, and control operations are known to conserve energy. If we wish to construct a multi-task model to learn the experimental behaviour, this conservation law translates to the desirable inductive bias 
\begin{align}
    \sum_{y,k}\mathcal{P}^k_\theta(y\vert \vec{x}) g_y = C \implies \frac{1}{N}\sum_{y,k}\mathcal{P}^k_\theta(y\vert \vec{x}) p_y = \frac{C}{NG}
\end{align}
where $G=\sum_{y}g_y$, $p_y=g_y^k/G$. Similarly to the examples above, this allows us to identify an operational equivalence that leads to a constraint
\begin{align}
\sum_{k} \sum_{y}\xi_{\Eb_{y\vert k}}( \lambda)p_y = \frac{C}{G} \quad\forall \;\lambda,
\end{align}
on a noncontextual learning model. 

\subsection{Possible extensions to our framework}
A significant challenge of this work involved deciding on how to cast machine learning in the language of contextuality, and the framework we present here is by no means the only option, nor the most general framework possible. For example, it is possible to define a notion of contextuality for transformations \cite{gencon}, in addition to the preparations and effects we consider in this work. This notion of contextuality may be more suited, for example, to studying the non-classically of layers in models with a sequential structure. Another option would be to consider a definition of contextuality that is defined with respect to the full model class, rather than a single model in the class as in Definition \ref{def:opcontmodel}. Doing this, one could study operational equivalences across different models in the class, although care would be needed to connect these equivalences to relevant concepts in learning. It is also worth stating clearly that the noncontextual constraints we consider in this work are qualitatively different from the constraints in Kochen-Specker type contextuality \cite{kscon}, which is studied in most works connecting quantum computation to contextuality. Kocken-Specker type contextuality can still exist within our framework, but follows from qualitatively different operational equivalences (that we briefly described in Section \ref{sec:gencon}). For this reason our work is something of a departure from the usual perspective of contextuality and could potentially find applications in the wider field of quantum computation.

Generalised contextuality has also recently been connected to the framework of generalised probabilistic theories \cite{gpt1,gpt2,gptcon1,gptcon2} (see also \cite{lami} for an application to associative memories). This framework is commonly used in foundations research and allows one to transform any probabilistic theory, such as classical and quantum theory, into a simple linear structure in which the states and effects are represented as vectors in a vector space\footnote{this is essentially what we have done in the geometric picture of section \ref{sec:linking} and Figure \ref{fig:thmfig}.}. In the generalised probability theory framework, noncontextuality becomes equivalent to a simple geometric condition, called simplex embedability \cite{gptcon1}, which relates to the possibility of embedding the state space into a higher dimensional simplex, and the effect space into its dual. Understanding this structural limitation in the context of learning models may therefore lead to general insights about the structural difference between contextual and noncontextual learners, and could potentially be used to define a measure of contextuality that goes beyond the binary contextual/noncontextual framework we assume.

\section{Outlook}\label{sec:outlook}
Although quantum machine learning has progressed in many ways, it nevertheless remains unclear why a machine learning practitioner would reach for the quantum toolbox amongst the vast array of classical tools on offer. In order to transform quantum machine learning into a science that has real-world impact, it is therefore crucial that we identify the features of quantum models that make them a powerful paradigm for learning. Historically, this power has often been viewed through the lens of quantum computation, and this has led to a program of searching for quantum advantage in terms of asymptotic complexity-theoretic analysis. 

Whereas it is true that some form of computational hardness is required to go beyond the capabilities of classical models, taking this perspective as a starting point may be overly restrictive due to the general difficulty of proving complexity-theoretic separations, and the highly specific and complex  mathematical structures they use. Our work offers a fresh perspective however: perhaps an alternative way to view the value of quantum machine learning is not in terms of computation, but through the ability to parameterise expressive model classes that encode certain data structures. As we have shown, contextuality offers an interesting perspective with which to view this question, and has led us to a specific data structure that can be naturally encoded in quantum models. In order to make further progress in this direction, we believe that greater effort is needed to \emph{
\begin{enumerate}[i]
    \item Identify mathematical structures that are characteristic of quantum theory and understand how these connect to structures present in data; 
    \item Determine which of these structures are connected to non-classicality and thus are likely to be difficult to emulate classically;
    \item Understand the kinds of learning problem where these non-classical structures plays a useful role.
\end{enumerate}}
The nascent field of geometric quantum machine learning \cite{geoml1,geoml2,geoml3}---the study of encoding symmetry structures into quantum models---may be a useful framework in which to tackle the first of these points. 
Other insights may also be possible by borrowing ideas from the quantum foundations literature, particularly for point (ii). For example, how does our framework of contextuality (or extensions thereof) connect to the notion of invariance and equivariance in geometric machine learning? What other common structures in contextuality (such as Kochen-Specker sets \cite{kscon} and antidistinguishability \cite{antid}) can be connected to biases in data? Is the non-signalling multi-partite structure of Bell nonlocality useful for learning certain data sets? Do concepts in quantum foundations suggest how to build a non-classical `quantum neuron'? Answers to these questions will generally not lead to statements of quantum advantage in the language of complexity theory. What they reveal in terms of understanding and intuition may however be valuable in the search for practical uses for quantum machine learning models. 

Perhaps the most pertinent question regarding this particular work is whether the linearly conserved quantities we consider (for potentially larger label dimension) can be usefully encoded in classical models using a tractable amount of resources. The fact that neural networks are nonlinear functions of the input data is actually problematic here, since it becomes difficult to impose linear constraints on the nonlinear output variables \cite{linml}. The linear structure of quantum theory is particularly useful however, as it provides a model in which the output label probabilities are nonlinear in the input data, but the linear bias is nevertheless enforced. Furthermore, classical mechanics is already used as a framework to encode conserved quantities in continuous-valued data \cite{hamnn, lagnn}. Since quantum theory can be seen as a generalisation of classical mechanics that deals with discrete variables (quanta), this does suggest that quantum theory is a very natural structure in which to encode conserved quantities of discrete data. 

\section {Acknowledgements}
We thank Richard East, Gillermo Alonso, Chae-Yeun Park and Flavien Hirsch for useful discussions, and Evan Peters, Richard East and Korbinian Kottmann for comments on the draft. This  project  has  received  funding  from  the  European  Union's  Horizon~2020  research and innovation programme under the Marie Sk\l{}odowska-Curie grant agreement No.~754510. VJW and MF acknowledge support from the Government of Spain (FIS2020-TRANQI, Severo Ochoa CEX2019-000910-S), Fundaci\'o Cellex, Fundaci\'o Mir-Puig and Generalitat de Catalunya (CERCA, AGAUR SGR 1381).
\printbibliography

\appendix

\section{The problem with using the physical theory as an operational scenario}
The most common perspective on generalised contextuality is that it is a property of a physical theory. This makes sense if your motivation is to answer the question `is quantum theory contextual?', however one runs into problems if this perspective is carried over into potential definitions of contextuality within machine learning where the focus is not on a specific theory but on the capabilities of a learning model. To illustrate this issue, consider the following two multi-task models, where we consider the domain $D_x$ of possible inputs to be a finite. 

\begin{enumerate}\label{app:def}
\item A quantum multi-task model where the preparations and effects correspond to density matrices $\rho_{\vec{x}}$ and measurement operators $M_{y}^k$ acting on $\mathcal{H}_d$.
\item A quantum multi-task model that simulates the above model exactly using diagonal density matrices $\tilde{\rho}_{\vec{x}}$ and diagonal measurement operators $\tilde{M}_{y}^k$ acting on $\mathcal{H}_{d'}$ with $d'\geq d$. 
\end{enumerate}
The second model always exists since any quantum model can be simulated with enough classical resources, and this classical simulation can be encoded into the computational basis states of a quantum system of larger dimension. In particular, one can use the choice $\tilde{\rho}_{\vec{x}}=\ket{i(\vec{x})}\bra{i(\vec{x})}$, $i(\vec{x})\in\{1,\cdots,\vert D_x \vert\}$ (so that the inputs are uniquely indexed by computational basis states), and $\tilde{M}_y^k = \sum_i \text{tr}(\rho_{\vec{x}(i)}M_y^k)\ket{i}\bra{i}$. One finds $\text{tr}(\tilde{\rho}_{\vec{x}}\tilde{M}_y^k)=\text{tr}({\rho}_{\vec{x}}{M}_y^k)$ as desired. 

If one takes the perspective that the operational scenario should correspond to the underlying theory (in this case, quantum theory), then one will arrive at different conclusions regrading the contextuality of the above two models. This is most easily seen by noting that there can be no preparation equivalences in model 2, since any diagonal state has a unique decomposition in terms of computational basis states. In the absence of preparation equivalences it is impossible to have generalised contextuality, as can be shown by constructing an explicit noncontextual ontological model where the ontic space corresponds to the set of density matrices, i.e. $\Lambda = \{\rho\}$, and the measurement response functions are given by the Born rule. It follows that the model 2 can never exhibit contextuality, whereas model 1 in principle can due to existing proofs of generalised contextuality in finite dimensional quantum theory.  

If we take the underlying physical theory to define the operational scenario we are therefore forced to accept that our notion of contextuality should depend on physical details that are entirely irrelevant for the actual behaviour of the model. This is unreasonable in our opinion since it doesn't allow one to connect contextuality to the behaviour of the device; hence the route we propose. 

\section{Proof of Theorem \ref{thm:main} for $\eta=1$}\label{app:zeronoise}
Recall from that 
\begin{align}
    \vec{v}_{\vec{x}} = (\mathcal{P}_\theta^1(+1\vert \vec{x}),\mathcal{P}_\theta^2(+1\vert \vec{x}),\mathcal{P}_\theta^3(+1\vert \vec{x}))=(P(\Eb^1_{+} \vert\Sb_{\vec{x}}),P(\Eb^2_+ \vert\Sb_{\vec{x}}),P(\Eb^3_+ \vert\Sb_{\vec{x}}))
\end{align} 
and $V = \{\vec{v}_{x}\vert \vec{x}\in D_{\vec{x}}\}$.
Since we have $\eta=1$, $V$ contains the points
\begin{align}\label{nonoisepoints}
    \vec{v}_1 = (1,0,\frac{1}{2}),\; \vec{v}_2 = (0,1,\frac{1}{2}),\;   \vec{v}_3 = (\frac{1}{2},1,0),\;  \vec{v}_4 = (\frac{1}{2},0,1),\; \vec{v}_5 = (0,\frac{1}{2},1), \;  \vec{v}_6 = (1,\frac{1}{2},0).
\end{align}
It follows that there exist inputs $\vec{x}_i$ such that $\vec{v}_{\vec{x}_i}=\vec{v}_i$. Since the vectors $\vec{v}_i$ are extremal in the space of models that satisfy the bias, any $\vec{v}_{x_i}$ can be written as a convex combination of them. Thus, any preparation $\Sb_{\vec{x}}$ is operationally equivalent to a convex mixture of the six preparations $\Sb_{i}\equiv\Sb_{x_i}$,
\begin{align}\label{opeqv_prep2}
    \Sb_{\vec{x}} \sim \sum_i p_i(\vec{x})\Sb_{i} \quad \forall \vec{x}
\end{align}
for some convex weights $p_i(\vec{x})$.

We now consider an ontological model of the operational scenario. Assuming operational noncontextuality, the operational equivalence \eqref{opeqv_prep} implies
\begin{align}\label{mixtures}
   \frac{1}{2}( \mu_{\Sb_1}(\lambda)+ \mu_{\Sb_2}(\lambda)) =  \frac{1}{2}( \mu_{\Sb_3}(\lambda)+ \mu_{\Sb_4}(\lambda)) =  \frac{1}{2}( \mu_{\Sb_5}(\lambda)+ \mu_{\Sb_6}(\lambda)) \quad \forall \lambda.
\end{align}
Using
\begin{align}
    \supp{\frac{1}{2}( \mu_{\Sb_i}+ \mu_{\Sb_j})} =  \supp{\mu_{\Sb_i}}\cup \supp{\mu_{\Sb_j}}
\end{align}
it follows from \eqref{mixtures} that
\begin{align}\label{supps}
   \supp{\mu_{\Sb_1}} \cup \supp{\mu_{\Sb_2}} = \supp{\mu_{\Sb_3}} \cup \supp{\mu_{\Sb_4}} = \supp{\mu_{\Sb_5}} \cup \supp{\mu_{\Sb_6}}, 
\end{align}
where $\text{supp}$ is the support of the ontic distribution; i.e.\ the set of $\lambda$ which have non-zero probability. The set of ontic states is the union of possible ontic states for each strategy,
\begin{align}
    \Lambda = \bigcup_{\vec{x}\in D_{\vec{x}}} \supp{\mu_{\Sb_{\vec{x}}}}.
\end{align}
From \eqref{opeqv_prep2} it follows 
\begin{align}
    \supp{\mu_{\Sb_{\vec{x}}}}\subseteq \bigcup_i \supp{\mu_{\Sb_i}},
\end{align}
and so from \eqref{supps}
\begin{align}
     \Lambda = \bigcup_i \supp{\mu_{\Sb_i}} = \supp{\mu_{\Sb_1}} \cup \supp{\mu_{\Sb_2}} = \supp{\mu_{\Sb_3}} \cup \supp{\mu_{\Sb_4}} = \supp{\mu_{\Sb_5}} \cup \supp{\mu_{\Sb_6}}.
\end{align}
Next, consider the first elements of $\vec{v}_1$ and $\vec{v}_2$. These imply
\begin{align}
    P(\Eb^1_+\vert \Sb_1) = 1, \quad P(\Eb^1_+\vert \Sb_2) = 0.
\end{align}
Since 
\begin{align}
    \int_{\lambda \in \supp{\mu_{\Sb_1}}} \text{d}\lambda\, \mu_{\Sb_1}(\lambda)=1
\end{align}
if follows from
\begin{align}
    P(\Eb^1_+\vert \Sb_1) = \int_{\lambda \in \supp{\mu_{\Sb_1}}}\text{d}\lambda\, \mu_{\Sb_1}(\lambda)\xi_{\Eb^1_+}(\lambda) = 1
\end{align}
that $\xi_{\Eb^1_+}(\lambda)=1$ for $\lambda\in\supp{\mu_{\Sb_1}}$. In a similar fashion
\begin{align}
    P(\Eb^1_+\vert \Sb_2) = \int_{\lambda \in \supp{\mu_{\Sb_2}}}\text{d}\lambda\, \mu_{\Sb_2}(\lambda)\xi_{\Eb^1_+}(\lambda) = 0
\end{align}
implies $\xi_{\Eb^1_+}(\lambda)=0$ for $\lambda\in\text{supp}(\mu_{\Sb_2})$. The ontic response function $\xi_{\Eb^1_+}(\lambda)$ is therefore deterministic on $\supp{\mu_{\Sb_1}}\cup\supp{\mu_{\Sb_1}}=\Lambda$. Analogous arguments using the pairs $\vec{v}_1$ and $\vec{v}_2$, or $\vec{v}_3$ and $\vec{v}_4$ can be constructed to show that $\xi_{\Eb^2_+}(\lambda)$ and $\xi_{\Eb^3_+}(\lambda)$ are also deterministic on $\Lambda$. We therefore have 
\begin{align}\label{detfuns}
    \xi_{\Eb_k^+}(\lambda) \in \{0,1\} \quad \forall \lambda\in \Lambda. 
\end{align}
We now turn to the operational equivalence \eqref{opeqv_meas}. Applying noncontextuality, this implies
\begin{align}
   \frac{1}{3}( \xi_{\Eb^1_+}(\lambda)+\xi_{\Eb^2_+}(\lambda)+\xi_{\Eb^3_+}(\lambda)) = \frac{1}{3}( \xi_{\Eb^1_-}(\lambda)+\xi_{\Eb^2_-}(\lambda)+\xi_{\Eb^3_-}(\lambda)) \quad \forall \lambda,
\end{align}
or using $ \xi_{\Eb_k^-}(\lambda) =  1-\xi_{\Eb_k^+}(\lambda)$,
\begin{align}
    \xi_{\Eb^1_+}(\lambda)+\xi_{\Eb^2_+}(\lambda)+\xi_{\Eb^3_+}(\lambda) = \frac{3}{2} \quad\forall \lambda.
\end{align}
Since we have proven that the functions $\xi_{\Eb_k^+}$ are deterministic in \eqref{detfuns}, there is no way of satisfying these equations for any $\lambda$. This contradiction therefore proves impossibility of a noncontextual learning model.

\section{Example of a model without preparation equivalences}\label{app:example}
As an example of such a model, consider a contextual multi-task model $\{\mathcal{P}_{\theta}^1(y\vert\vec{x}),\mathcal{P}_{\theta}^2(y\vert\vec{x}),\mathcal{P}_{\theta}^3(y\vert\vec{x})\}$. Construct a joint model from this as $\mathcal{P}_{\theta}^1(\vec{y}\vert\vec{x})=\mathcal{P}_{\theta}^1(y^{(1)}\vert\vec{x})\mathcal{P}_{\theta}^2(y^{(2)}\vert\vec{x})\mathcal{P}_{\theta}^3(y^{(3)}\vert\vec{x})$, i.e., the labels of the joint model are sampled independently using the multi-task model. Let us assume that the multi-task model is contextual in the sense of Theorem \ref{thm:main} for $\eta=1$, so that the six preparations densities $s_i$ correspond to the six vectors $\vec{v}_i$. From \eqref{vprobs} if follows that for the preparation density $\frac{1}{2}s_1+\frac{1}{2}s_2$ all labels can occur in the joint model except the labels $\vec{y}=(0,0,1),(1,1,1)$. Similarly for $\frac{1}{2}s_3+\frac{1}{2}s_4$ the labels $\vec{y}=(1,0,0),(1,1,1)$ never occurs, and for $\frac{1}{2}s_5+\frac{1}{2}s_6$ the labels $\vec{y}=(0,1,0),(1,1,1)$ never occur. This means that the label distributions for the different pairs of preparations densities are not the same, and hence the pairs of preparation densities are not operationally equivalent.

\section{Higher dimensional measurements satisfying the zero-sum bias condition}\label{matemagic}
Here we give an example of a set of three $\pm1$ valued observables $\mathcal{O}_k$ in even dimension $d$ that satisfy the zero-sum bias condition $\mathcal{O}_1+\mathcal{O}_2+\mathcal{O}_3=0$ for all even $d \ge 2$. We write the three observables as $\mathcal{O}_1=2P-\mathbb{I}$ and $\mathcal{O}_2=2Q-\mathbb{I}$ (and $\mathcal{O}_3 = -\mathcal{O}_1-\mathcal{O}_2$), where $P$ and $Q$ are the projections onto the positive eigenspaces of the observables $\mathcal{O}_1$ and $\mathcal{O}_2$, respectively. Since $\mathcal{O}_3$ is a $\pm1$ valued observable it follows that $\mathcal{O}_3^2=\mathbb{I}$, that is,
\begin{equation}
\mathbb{I} = (\mathcal{O}_1 + \mathcal{O}_2)^2 = 4(P + Q - \mathbb{I})^2 = 4(P+Q)^2 + 4\mathbb{I} - 8(P+Q) = 4(P + Q + PQ + QP) + 4\mathbb{I} - 8(P+Q)
\end{equation}
and therefore,
\begin{equation}
\frac34 \mathbb{I} = P+Q-PQ-QP = (P-Q)^2.
\end{equation}
The problem is therefore equivalent to finding projections $P$ and $Q$ such that $(P-Q)^2=\frac{3}{4}\mathbb{I}$. 

To achieve this for even dimensional Hilbert spaces, we write the Hilbert space as $\mathbb{C}^d \simeq \mathbb{C}^{\frac{d}{2}} \oplus \mathbb{C}^{\frac{d}{2}}$, and define $P$ and $Q$ in block forms
\begin{equation}
P =
\begin{pmatrix}
\mathbb{I} & 0 \\
0 & 0
\end{pmatrix},
\quad
Q = \frac14
\begin{pmatrix}
\mathbb{I} & \sqrt{3} \mathbb{I} \\
\sqrt{3} \mathbb{I} & 3 \mathbb{I}
\end{pmatrix}.
\end{equation}
It is straightforward to verify that $P$ and $Q$ are projections and that $(P-Q)^2 = \frac34 \mathbb{I}$ holds. Therefore, the corresponding observables satisfy the desired bias.


\section{Details of the unbiased quantum model}\label{app:ansatz}

\begin{figure}
    \centering
    \includegraphics[scale=0.9]{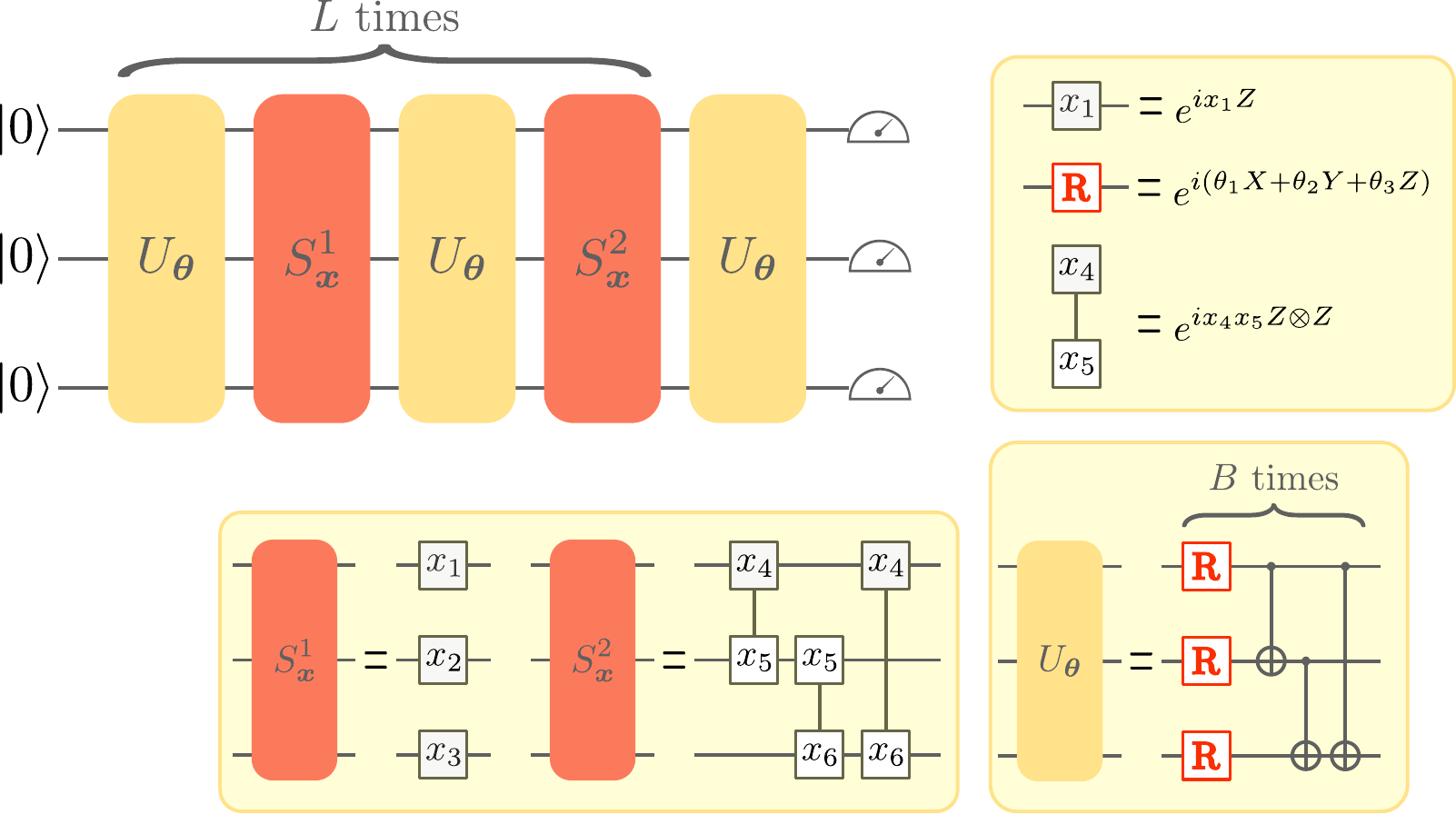}
    \caption{The unbiased quantum model. As for the biased qunatum model, each unitary $U(\vec{\theta})$ has 9 parameters given by the three parameterised qubit rotations. The data encoding unitaries are identical to the biased quantum model, and so this model has access to the same Fourier spectrum.}
    \label{fig:unst}
\end{figure}

The generic unbiased quantum model that we use for the numerical investigation of section \eqref{sec:surrogate} is described in Figure \ref{fig:unst}. 

\section{Data generation and training details}\label{app:numerics}
Here we outline the specific details pertaining to the numerical investigation of section \ref{sec:surrogate}. To generate each row of the $3\times 3$ strategy matrix $\vec{x}$, we sample three uniformly random numbers $r_1, r_2, r_3$ in $[0,1]$ and set the row entries as $(r_1,r_2,r_3)/\sum_i r_i$, which results in a valid probability distribution. This process defines the distribution $\mathcal{D}(\vec{x})$ described in the main text.  We repeat this 1500 times to generate all input strategies in the training set. The corresponding labels are sampled according to the rules of the RPS game via \eqref{conprior}. 

Both quantum models are initialised by sampling each of their parameters uniformly in $[0,2\pi]$. Each parameter of the surrogate model is initialised by sampling a number uniformly in $[-1,1]$ then multiplying by $1/\#\texttt{params}$ where $\#\texttt{params}=2(2L+1)^6$ is the total number of parameters of the model. To optimise the models we use JAX and the optax package, in combination with pennylane in the case of the quantum models. For the gradient descent, we employ full batch gradient descent (no significant different was observed using mini-batches). For the quantum models we use the adam update with the optax default settings and an initial learning rate of $0.001$. For the suroogate model, since the L2 regularised cost function is convex in the optimisation parameters, we used standard gradient descent with a learning rate of $0.001$ (L=1) and $0.0001$ (L=2). 

\begin{figure}
    \centering
    \includegraphics[width=\textwidth]{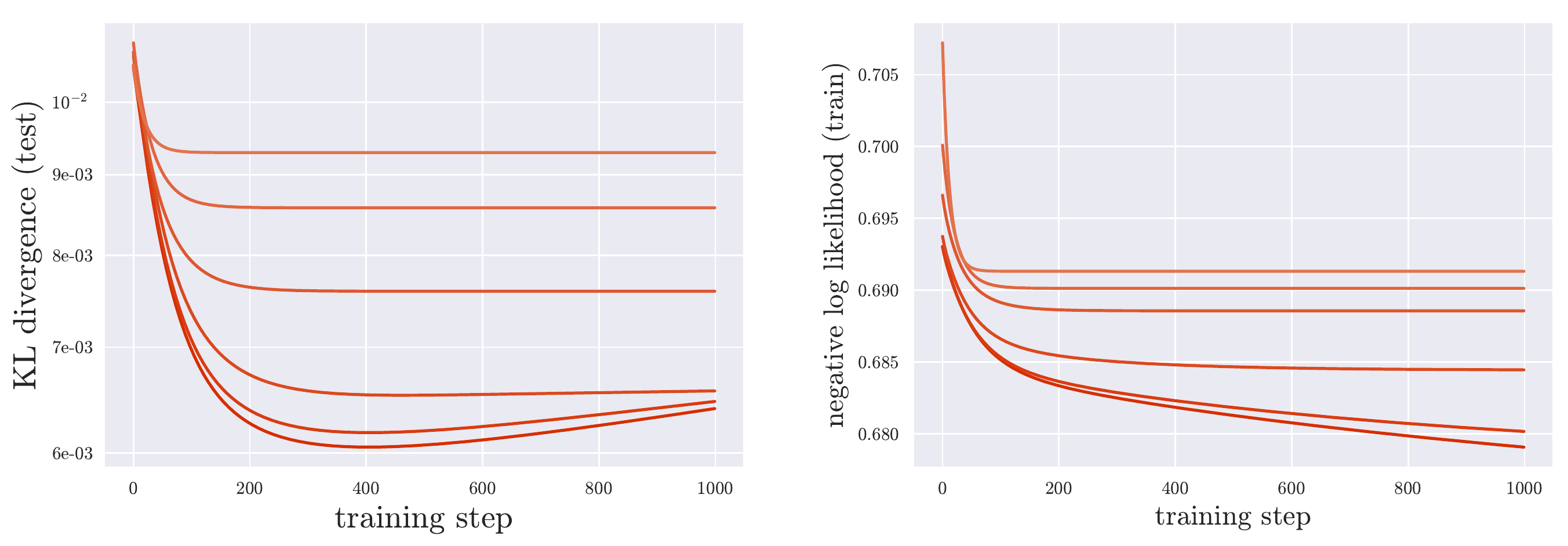}
    \caption{The effect of adding L2 regularisation to the cost function for the surrogate model for $L=1$ for different regularisation strengths (lighter curves imply stronger regularisation).}
    \label{fig:reg1}
\end{figure}

\end{document}